\documentclass{article}
\usepackage{natbib}
\usepackage{amsmath}
\usepackage{theorem}
\usepackage{macros}
\usepackage{graphicx}
\usepackage{xspace}
\usepackage[all]{xy}
\xyoption{rotate}
\usepackage{macros}
\usepackage{proof}
\usepackage{bold-extra}
\usepackage{ifthen}

\begin{document} 
\author{Miko{\l}aj Boja\'nczyk\footnote{Supported by the European Research Council (ERC) under the
European Union Horizon 2020 research and innovation programme (ERC consolidator grant LIPA, agreement no.
683080). }, University of Warsaw}
\title{Two monads for graphs}
\maketitle

\begin{abstract}
	An introduction to algebras for graphs, based on Courcelle's algebras of hyperedge replacement and vertex replacement. The paper uses monad notation. 
\end{abstract}

\tableofcontents
\pagebreak

\section{Introduction}
From the perspective of this paper, the basic idea behind a monad is that it defines a notion of ``structure'' (e.g.~word, tree, graph, etc.) and says how bigger structures can be composed from smaller structures.  This paper describes two monads that model graphs. 

 For finite words, it is easy to compose smaller words to get bigger words:  one uses word concatenation. For structures more complicated than finite words, such as trees or graphs, more care is needed, and one is typically forced to introduce bookkeeping decoration, such as extending trees or graphs with ``sources'' or ``interfaces''. There are numerous ways to model such sources, e.g.~a source could be a distinguished vertex, a distinguished hyperedge, a distinguished set of vertices, etc. The two monads for modelling graphs that are considered in this paper -- called $\hmonad$ and $\vmonad$ --  make different design decisions in this respect. The monads are inspired by two algebras for graphs called ``hyperedge replacement'' and ``vertex replacement'' algebras\footnote{
 The algebra for hyperedge replacement was introduced in
 \cite{Courcelle:1990fj}. For t	he origins of the  algebra for vertex replacement, see the discussion at the end of Chapter 2 in the book of \cite{Courcelle:2012wq}.
Throughout this paper, the book is referenced whenever possible  instead of the original papers.}.

The first purpose of this paper is to  present the results about these algebras using the terminology of monads, hoping that this presentation makes the ideas more accessible and easier to draw.  A second purpose is to develop some general theory about monads, motivated by the challenges presented by graphs.

To illustrate the monad method, we include a new result, Theorem~\ref{thm:decide-mso-not-count}, which characterises those languages of graphs of bounded treewidth that can be defined in \mso without counting (as opposed to  languages that can be defined in \mso with counting). The characterisation is effective, i.e.~there is an algorithm which inputs a sentence of counting \mso (and a guarantee that the sentence has models of bounded treewidth), and decides if the sentence is equivalent to one that does not use counting, see Theorem~\ref{thm:decide-mso-not-count}. In the process of describing the algorithm, we are forced to develop some theoretical infrastructure, such as how algebras can be represented and manipulated by algorithms. (The difficulty is that the algebras have infinitely many sorts and infinitary operations.)

\section{Monads}

%
This  section introduces basic definitions and notation for monads. Although monads are a concept from category theory, the paper is intended to be readable for readers without a background in category theory.  
The reader familiar with monads can skip Section~\ref{sec:monads-algebras} and go directly to Section~\ref{sec:polynomials}; the only difference between this paper and say, the wikipedia page about monads, is that we write $\unitt$ instead of $\eta$ and $\flatt$ instead of $\mu$\footnote{This notation is based on \cite{Anonymous:2015vr}
}. Section~\ref{sec:polynomials}  discusses the less standard notion of  polynomials for Eilenberg-Moore algebras. 
\subsection{Monads and their algebras}
\label{sec:monads-algebras}
 We use only the most basic concepts about monads. 
\begin{definition}[Monad]
 A monad consists of four ingredients, as given in Figure~\ref{fig:monad-ingredients}, subject to six axioms, as given in Figure~\ref{fig:monad-axioms}.	
\end{definition}
 
\begin{figure}
\newcommand{\tmonadline}[3]{(#1) &
\begin{minipage}{0.5\textwidth}
#2
\end{minipage} &
\begin{minipage}{0.4\textwidth}
#3
\end{minipage}
\\ \\ \hline  \\
}
\begin{tabular}{lll}
& Ingredient of a monad & Example for  finite words \\ 
\hline \\
\tmonadline 1 {A map from sets to sets, which takes  each set $\Sigma$ a new set  
  $\tmonad \Sigma$. The
   intuition  is that $\tmonad \Sigma$ 
   represents ``structures'' labelled by $\Sigma$.}
   {The structures are finite words, i.e.~$\tmonad \Sigma = \Sigma^*$}
\tmonadline 2 {A map from functions to functions, which lifts  each function  
\begin{align*}
\Sigma \stackrel f \to \Gamma	
\end{align*}
 to a  function on structures  
 \begin{align*}
\tmonad \Sigma \stackrel {\tmonad f} \to \tmonad \Gamma.
\end{align*}
}{ The function~$\tmonad f$ replaces each letter by its image under~$f$}

\tmonadline3 {For each set $\Sigma$, a function
\begin{align*}
\unitt_\Sigma : \Sigma \to \tmonad \Sigma.
\end{align*}
Intuitively speaking, $\unitt$ says how letters can be interpreted as structures.
}{A letter $a \in \Sigma$ is mapped to a one letter word $a \in \Sigma^*$.}

\tmonadline 4{For each set $\Sigma$, a function
\begin{align*}
\flatt_\Sigma : \tmonad(\tmonad \Sigma) \to \tmonad \Sigma.
\end{align*}
Intuitively speaking, $\flatt$ says how a structure of structures can be flattened to a structure.
}{A word of words is flattened to a word, like
\begin{align*}
	(aba)(aa)(\varepsilon)(a) \quad \mapsto \quad abaaaa
\end{align*}
}
\end{tabular}  \caption{\label{fig:monad-ingredients}The ingredients of a monad in the category of sets. The right column shows how these ingredients are instantiated  for the monad of finite words.}
\end{figure}

	\begin{figure}[htbp]
		\centering
	\begin{align*}
		\vcenter{\xymatrix{ \tmonad \Sigma \ar[r]^{\tmonad f} \ar[dr]_{\tmonad(g \circ f)}& \tmonad \Gamma \ar[d]^{\tmonad g}  \\ & \tmonad \Delta}
} \qquad 		\vcenter{\xymatrix @R=2pc { \tmonad  \Sigma  \ar@<-.5ex>[r]_{\mathrm{id}_{\tmonad \Sigma}} \ar@<.5ex>[r]^{\tmonad (\mathrm{id}_\Sigma)}  & \tmonad \Sigma}}
	\end{align*}
	\begin{align*}
		\vcenter{\xymatrix @R=2pc { \Sigma  \ar[r]^{f} \ar[d]_{\unitt_\Sigma} & \Gamma \ar[d]^{\unitt_\Gamma}  \\
		\tmonad \Sigma \ar[r]_{\tmonad f}& \tmonad \Gamma
		}} \qquad 		\vcenter{\xymatrix @R=2pc { \tmonad \tmonad \Sigma  \ar[r]^{\tmonad \tmonad f} \ar[d]_{\flatt_\Sigma} & \tmonad \tmonad \Gamma \ar[d]^{\flatt_\Sigma}  \\
		\tmonad \Sigma \ar[r]_{\tmonad f}& \tmonad \Gamma
		}}.
	\end{align*}
	\begin{align*}
		\vcenter{\xymatrix @R=2pc @C=3pc {\tmonad \tmonad \tmonad \Sigma  \ar[r]^{\flatt_{ \tmonad \Sigma}} \ar[d]_{\tmonad{ \flatt_\Sigma}} & \tmonad \tmonad \Sigma \ar[d]^{\flatt_\Sigma}  \\
		\tmonad \tmonad \Sigma \ar[r]_{ \flatt_\Sigma}& \tmonad \Sigma
		}} 
		\qquad
		 		\vcenter{\xymatrix @R=2pc { \tmonad \Sigma \ar[dr]^{\mathrm{id}_\Sigma}  \ar[r]^{\unitt_ {\tmonad \Sigma}} \ar[d]_{\tmonad \unitt_\Sigma} & \tmonad \tmonad \Sigma \ar[d]^{\flatt_\Sigma}  \\
		\tmonad \tmonad X \ar[r]_{ \flatt X }& \tmonad \Sigma
		}}
	\end{align*}
		\caption{The axioms of a monad are that these six diagrams  commute for every set $\Sigma$  and  every functions $f : \Sigma \to \Gamma$ and $g : \Gamma \to \Delta$. The diagrams in the first row say that $\tmonad$ is a functor.  The diagrams in the middle row say that the unit and flattening are natural. The lower left diagram says that  flattening is associative, and the lower right says that the unit is consistent with flattening.} 
		\label{fig:monad-axioms}
	\end{figure}
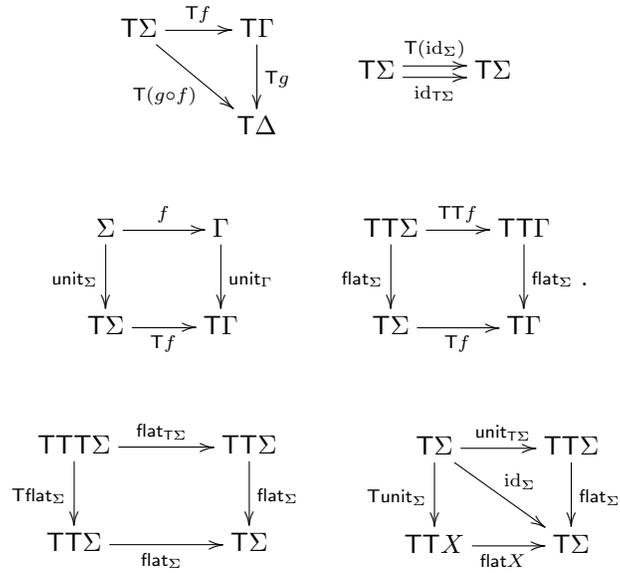

Figure~\ref{fig:monad-ingredients} also illustrates the monad ingredients for the monad of finite words. 
The description in Figures~\ref{fig:monad-ingredients} and~\ref{fig:monad-axioms} uses the  category of sets, but monads can be defined in other categories (by replacing  ``set'' with ``object in the category'' and replacing ``function'' with ``morphism in the category''). Actually, the two monads $\hmonad$ and $\vmonad$ discussed at length in this paper are \emph{not} in the category of sets, but in categories of  sorted sets (with the sorts being $\Nat$ and $\Nat - \set 0$, respectively). By a category of sorted sets, we mean any category of the form $\setcat^X$ for some set $X$ called the \emph{sort names}. In such a category, the objects are sets where each element has an associated sort name from $X$, and the morphisms  are sort preserving functions. When there is only one sort name, we recover the category of sets.  


\paragraph*{Algebras and the languages that they recognise.} Every monad comes with an associated notion of algebra.
\begin{definition}
	[Algebras  and homomorphisms] Let $\tmonad$ be a monad. 
	A \emph{$\tmonad$-algebra}, also known as an Eilenberg-Moore algebra over the monad, is defined to be a morphism (in the category) of the form 
	\begin{align*}
\pi : \tmonad A \to A,
\end{align*}
	 where $A$ is called the \emph{universe} of the algebra, 
	such that $\pi \circ \unitt_A$ is the identity on $A$ and the following diagram commutes:
	\begin{align*}
		\xymatrix  { \tmonad \tmonad A  \ar[r]^{\flatt_A} \ar[d]_{\tmonad \pi} & \tmonad A \ar[d]^{\pi}  \\
		\tmonad A \ar[r]_{\pi} & A
		} 
	\end{align*}
A homomorphism between two $\tmonad$-algebras
	\begin{align*}
 \pi_{\alga}: \tmonad A \to A	 \qquad  \pi_{\algb}: \tmonad B \to B	\end{align*}
is any  morphism in the category  $h : A \to B$ which makes the following diagram commute:
		\begin{align*}
		\vcenter{\xymatrix  { \tmonad  A  \ar[r]^{\tmonad h} \ar[d]_{ \pi_\alga} & B \ar[d]^{\pi_\algb}  \\
		A \ar[r]_{h} & B
		}}
	\end{align*}
\end{definition}
For example,  when the monad is the monad of finite words described in the right column of Figure~\ref{fig:monad-ingredients}, then the algebras are monoids. A simple and important fact is that for every monad $\tmonad$ and every object $\Sigma$ in the category, $\tmonad \Sigma$ becomes an algebra if we take the product operation to be flattening; this is because of the axioms in the last row of Figure~\ref{fig:monad-axioms}.

\begin{definition}[Functions and languages recognised by algebras]
	Let $\tmonad$ be a monad and let
	\begin{align*}
f : \tmonad \Sigma \to C	
\end{align*}
be a morphism in the category, e.g.~a sort-preserving function when the category is sorted sets.  We say that $f$ is recognised by a $\tmonad$-algebra $\alga$ if there is a homomorphism $h$ and a morphism in the category $g$ which makes the following diagram commute
\begin{align*}
\xymatrix{
\tmonad \Sigma \ar[dr]_h \ar[r]^f & C\\
& \alga \ar[u]_g
}	
\end{align*}
If the category is sorted sets, then we say that a set
\begin{align*}
L \subseteq \tmonad \Sigma	
\end{align*}
is recognised by $\alga$ if the characteristic function of $L$ is recognised (the co-domain of the characteristic function is a sorted set which has two values ``yes'' and ``no'' on each sort).
\end{definition}
%

Suppose that a monad $\tmonad$ is equipped with a notion of ``finite algebra''; then  a \emph{recognisable morphism/language} is defined  to be one that is recognised  by a finite algebra.
 For example, if the monad is in the category of sets, then a natural notion of ``finite algebra'' is that the universe is a finite set. If the category is multisorted sets, then a natural notion of ``finite algebra'' is that the universe is finite on every sort. If the monad is in the category is vector spaces, then a natural notion of ``finite algebra''  is that the universe has finite dimension. In general, the notion of ``finite algebra'' has to be given as a parameter, and choosing the right one can be a non-trivial task (an important example is monads for infinite trees, where the notion of ``finite algebra'' is yet to be identified). 
%
%
\subsection{Polynomials}
\label{sec:polynomials}
 
In this section we describe {polynomial operations} in an Eilenberg-Moore algebra. This notion is less standard; and the presentation in this section assumes  that the monad is in a category of sorted sets. A  more relaxed and yet still sufficient assumption would be to use  a concrete category (i.e.~a category equipped with a faithful functor to the category of sets). Nevertheless, we restrict attention to  sorted sets, since these are the categories we use for the examples in this paper.

 Polynomial  operations as defined here are  based on polynomial  operations from universal algebra\myfootcite[Definition 13.3]{burris2006course}, which are in turn based on the usual notion of polynomials like $3x^2 + 2xy^2 - 3$.
Fix a category of sorted sets (for some sort names).
Consider a monad $\tmonad$ over this category, and an algebra $\alga$ over this monad whose universe is $A$. For a sorted set $X$, called the  \emph{variables}, define 
\begin{align*}
\mathsf{eval} :  \underbrace{A^X}_{\substack{\text{sort-preserving}\\ \text{functions from $X$ to $A$}}} \times \quad \tmonad (A + X) \quad  \to \quad   A	
\end{align*}
to be the operation which, on input $(\eta,t)$,  substitutes the variables in $t$  according to the valuation $\eta$, and then applies the product operation of the algebra $\alga$. We use the name \emph{variable valuations} for the first argument of the operation. The operation $\mathsf{eval}$ is not itself a morphism in the category, one reason being that the variable valuations $A^X$ do  not form a sorted set.  If we fix the first argument of the function $\mathsf{eval}$ to be some  variable valuation $\eta$, then we do get a morphism in the category, i.e.~a sort-preserving function
\begin{align*}
\mathsf{eval}(\eta,\_) : \tmonad (A + X)	 \to A.	
\end{align*}
We are, however, mainly interested in fixing the second argument of the operation $\mathsf{eval}$. For  $t \in \tmonad (A + X)$, define 
\begin{align*}
\sem t = \mathsf{eval}(\_,t)  : A^X \to A	,
\end{align*}
to be the function which inputs a variable valuation and outputs the result of applying it to $t$. Because the variable valuations do not form a sorted set, the function $\sem t$ is not sort-preserving (unless the category is sets, i.e.~there is only one sort name). Functions of the form $\sem t$ are called polynomial  operations, as described in the following definition.

\begin{definition}[Polynomial  operation]\label{def:polynomial} Let $\tmonad$ be a monad in a category of sorted sets, and let  $\alga$ be a $\tmonad$-algebra whose universe is $A$. For a sorted set $X$, a \emph{polynomial  operation over $\alga$ with variables $X$} is defined to be any function
\begin{align*}
f : A^X \to A	
\end{align*}
which is of the form $\sem t$ for some $t \in \tmonad (A+X)$. We write $\alga[X]$ for the set of polynomial operations over $\alga$ with variables $X$.  	
\end{definition}

A polynomial  operation with an empty set of variables is the same thing as an element of the algebra, i.e.~a constant.
If the category is sets, i.e.~there is only one sort,  then it is enough to say how many variables there are, and in such a situation  one usually speaks of unary polynomial operations, binary polynomial operations, etc. This is the case for the monad of finite words, where the algebras are monoids. For  algebras with more than one sort -- which will be the case  the monads $\hmonad$ and $\vmonad$ that are the topic of this paper -- one also needs to remember the sort for each variable, hence our definition which keeps track of the (sorted) set of variables.

\begin{myexample}
	Suppose that the monad is the monad of finite words over the category of finite sets. In this monad, the algebras are monoids.  Consider the free monoid $\set{a}^*$.  The function
	\begin{align*}
w \in \set{a}^* \mapsto waw	 \in \set{a}^*.
\end{align*}
is a polynomial  operation, namely $\sem { xax}$. 
In general, a polynomial  operation over the free monoid $\set{a}^*$ with variables  $\set{x}$ is given by a word over the alphabet $\set{a}^* + \set x$. 
\end{myexample}

\begin{definition}[Congruence]\label{def:congruence}
	An equivalence relation $\sim$ on the universe $A$ of an $\hmonad$-algebra is called \emph{compatible} with an operation
	\begin{align*}
f : A^X \to A	
\end{align*}
if for every two inputs that are equivalent under $\sim$ (equivalent on each variable $x \in X$), the outputs are also equivalent under $\sim$. A \emph{congruence} in an $\hmonad$-algebra is an equivalence relation on the universe of $\alga$ that is compatible with all polynomial operations.
\end{definition}

The following result, which is simply a monad version of a classic result from universal algebra,  says that congruences are essentially the same thing as homomorphisms.

\begin{theorem}[Homomorphism Theorem]\myfootcite[Lemma 3.3]{Anonymous:2015vr} An equivalence relation on the universe of a $\hmonad$-algebra is a congruence if and only if it is the kernel of some homomorphism.
\end{theorem}

%
%

\section{Hyperedge replacement}
\label{sec:hyperedge}
In this section, we describe the first of the two monads discussed in this paper, namely the  monad $\hmonad$, which is used to  model graphs and hypergraphs. 
The name of the monad  $\hmonad$ stands for \emph{hyperedge replacement}. The monad is based on Courcelle's algebra of hyperedge replacement. A second monad, called  $\vmonad$ and based on Courcelle's algebra of vertex replacement, will be discussed in Section~\ref{sec:vertex}.

\paragraph*{Hypergraphs.}
We begin by describing the  variant of hypergraphs that is used to define the monad $\hmonad$. Our real goal is to model graphs, but the more general model of  hypergraphs is needed to describe the compositional structure needed in a monad. The idea is  that a hyperedge incident with $n$ vertices can be replaced by a hypergraph with $n$ distinguished vertices. 
Before formally defining  hypergraphs, we review some  design decisions. The first design decision is:
\begin{itemize}
	\item Hyperedges are labelled.
\end{itemize}
     Labels are essential to the monad approach: hyperedge labels are used to define the substitution operation that defines the monad. The  three other decisions described below are  not essential, and one can define monads which make different design decisions in these respects.
\begin{itemize}     
	\item Hyperedges are non-looping, i.e.~a hyperedge is incident to a non-repeating list of vertices. 
	\item Hyperedges are directed, i.e.~a hyperedge is incident to a list and not a set of vertices. 
	\item Parallel hyperedges are allowed, i.e.~it is possible to have several hyperedges that have the same label and incidence list.
\end{itemize}

When defining the monad $\hmonad$, we use \emph{ranked sets}. A   \emph{ranked set} is a set where every element has an associated \emph{arity} in $\set{0,1,2,\ldots}$, i.e.~this is a special case of a sorted set, where the sorts are natural numbers.  The idea is that hyperedges and their labels will be ranked, the arity of a hyperedge being the length of its incidence list. 
An \emph{arity-preserving function} is a function between two ranked sets which does not change arities.

%

\newcommand{\bang}[1]{[#1]}
\begin{definition}[Hypergraph] \label{def:hypergraphs}
A \emph{hypergraph} consists of:
	\begin{itemize}
		\item a (not ranked) set $V$ of \emph{vertices};
		\item a ranked set $E$ of \emph{hyperedges};
		\item an \emph{incidence} function $E \to V^*$ which maps each $n$-ary-vertex to its \emph{incidence list}, which is non-repeating and has length $n$.		\item a ranked set $\Sigma$ of \emph{labels} together with a  arity-preserving \emph{labelling} $E \to \Sigma$;
	\end{itemize}
\end{definition}

We write $G,H$ for hypergraphs, $e,f$ for hyperedges and $v,w$ for vertices.
If $e$ is an $n$-ary hyperedge and $i \in \set{1,\ldots,n}$, then we write $e[i]$ for the $i$-th vertex incident to $e$, i.e.~for the $i$-th element of the incidence list of $e$.  
 Note that we allow  hyperedges of arity zero (called nullary) which have empty incidence lists.     Hypergraphs generalise directed graphs without self-loops, see Example~\ref{ex:graphs} below; and in fact modelling graphs is the main goal of the monad. 

\begin{myexample}\label{ex:graphs}
Consider a directed graph $G=(V,E)$ without self-loops, i.e.~$V$ is a set and $E \subseteq V \times V$ is a binary irreflexive relation. We view $G$ as a hypergraph in the following way. The vertices of the hypergraph are $V$, the  hyperedges are binary and correspond to edges. Each hyperedge has the same label, call it ``edge'',  of rank 2 and the incidence list of an edge consists of its source and target. 
\end{myexample}

 \paragraph*{Drawing hypergraphs.} Suppose that the set of labels $\Sigma$ has two elements as shown below:
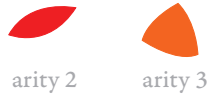
\mypic{8}
Here is a picture of a hypergraph labelled by $\Sigma$:
\mypic{7}  
To draw the incidence lists, which are ordered, we use the following convention:
\mypic{34}  

\paragraph*{Sourced hypergraphs.} To define the monad structure on hyperegraphs, we equip  hypergraphs with distinguished vertices, called sources. The idea is that an $n$-ary  hyperedge can be replaced by a hypergraph with $n$ distinguished vertices.

\begin{definition}[Sourced hypergraph] \label{def:sourced-hypergraph}
A  \emph{sourced hypergraph} is defined to be a hypergraph $G$ together with an injective function
\begin{align*}
\text{source} : \set{1,\ldots,n} \hookrightarrow \text{vertices of $G$}	
\end{align*}
for some $n \in \set{0,1,\ldots}$.  	The number $n$ is called the arity of the sourced hypergraph.
\end{definition} 
Sourced hypergraphs form  a ranked set. In the end, we are mainly interested in the nullary sourced hyperegraphs, i.e.~hypergraphs. Sourced hyperegraphs of higher arities are used to define the substitution operation. We draw sourced hypergraphs like this (the picture shows a binary sourced hypergraph):
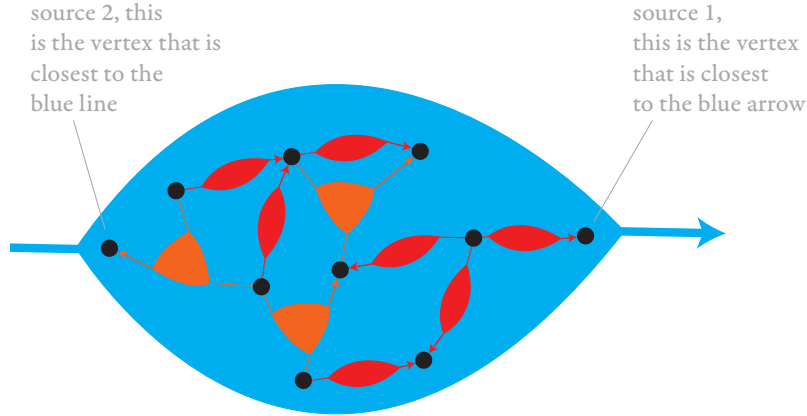
\mypic{9} 
	The point of the drawing is to underline a   sourced hypergraph, having an arity, can also be used  as a label in a hypergraph. 
	 We are now ready to define the first ingredient of the monad $\hmonad$, i.e.~how it works on objects (ranked sets).

\begin{definition}[The monad $\hmonad$]
	The monad $\hmonad$ is defined as follows.
	\begin{itemize}
		\item \emph{Category.} The category is ranked sets and arity-preserving functions, with the ranks being $\set{0,1,\ldots}$;
		\item \emph{On objects.} 	For a ranked set $\Sigma$, define $\hmonad \Sigma$ to be the ranked set of hypergraphs with labels $\Sigma$ that are finite\footnote{Finitely many vertices and hyperedges.}, modulo isomorphism\footnote{Two sourced hypergraphs are called \emph{isomorphic} if they have the same set of labels and   there are two bijections -- between the vertices and the hyperedges -- which respect the labelling, incidence, and source functions in the natural way. 

		 In the terminology of Courcelle and Engelfriet, $\hmonad \Sigma$ talks about \emph{abstract} hypergraphs (i.e.~hypergraphs up to isomorphism) as opposed to \emph{concrete} hypergraphs (i.e.~not up to isomorphism). The main reason for working modulo isomorphism is that we want to have certain ways of combining  hypergraphs to give equal results (e.g.~we want disjoint union of hypergraphs to be associative) and this can only be achieved modulo isomorphism.
}.
		\item \emph{On morphisms.} For an arity-preserving function $f : \Sigma \to \Gamma$, the function $\hmonad f : \hmonad \Sigma \to \hmonad \Gamma$ is defined by applying $f$ to the labels;
		\item \emph{Unit.} The unit of  an  $n$-ary label $a \in \Sigma$ is the $n$-ary sourced hypergraph with vertices $\set{1,\ldots,n}$,  one $n$-ary hyperedge labelled by $a$ and incident to $[1,\ldots,n]$, and where the source mapping is the identity. Here is a picture:
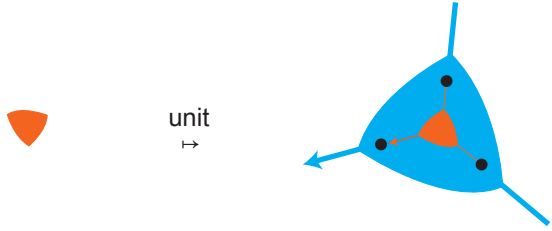
\mypic{11}
\item \emph{Flattening.} For a sourced hypergraph $G \in \hmonad \hmonad \Sigma$, its flattening is defined as follows (see Figure~\ref{fig:h-flattening}):
 \begin{itemize}
\item Hyperedges are pairs $(e,f)$, where $e$ is a hyperedge of $G$ and $f$ is a hyperedge in the label of $e$. The colour and arity of a hyperedge $(e,f)$ is inherited from $f$.
	\item Vertices are either vertices $v$ of $G$, or pairs $(e,v)$ such that $e$ is a hyperedge of $G$ and $v$ is a non-source vertex in the label of $e$. The source function is inherited from $G$.
\item The incidence list of a hyperedge $(e,f)$ is obtained by taking the incidence list of $f$ (in the label of $e$), and applying the following \emph{parent} function:
\begin{eqnarray*}
\text{parent}(v) &=& v\\ \qquad \text{parent}(e,v) &=& \begin{cases}
	(e,v) : \text{if $v$ is not a source in the label of $e$}\\
	e[i] : \text{if $v$ is the $i$-th source in the label of $e$}
\end{cases}	
\end{eqnarray*}
\end{itemize}
	\end{itemize}
\end{definition}

The above definition  restricts attention to finite sourced hypergraphs.  In principle, we could also consider infinite ones -- and in fact the monad approach would have its advantages for considering infinite objects, since it does not insist on any explicit use of finite terms to define objects. The main reason for finiteness is that we want to remain close to the approach from~\cite{Courcelle:2012wq}, which uses finite structures. Another advantage of finiteness is that the  operation $\hmonad \Sigma$ is \emph{finitary} in the following sense: 
\begin{align*}
\hmonad \Sigma = \bigcup_{\Sigma_0 \subseteq \Sigma} \hmonad \Sigma_0
\end{align*}
where the union above ranges over finite subsets $\Sigma_0$. The fact that $\hmonad$ is finitary will play a role in Section~\ref{sec:mso}.

\begin{figure}
\mypic{3}
 \mypic{6}
  \caption{\label{fig:h-flattening}A sourced hypergraph in $\hmonad \hmonad \Sigma$ (above) and its flattening (below).}
\end{figure}

\begin{fact}\label{fact:hmonad-is-amonad}
	$\hmonad$ satisfies the axioms of  a monad.
\end{fact}
\begin{proof}A routine check of the axioms. We only check associativity, i.e.~that the following diagram commutes
\begin{align*}
	\vcenter{\xymatrix @R=2pc @C=3pc {\hmonad \hmonad \hmonad \Sigma  \ar[r]^{\flatt_{ \hmonad \Sigma}} \ar[d]_{\hmonad{ \flatt_\Sigma}} & \hmonad \hmonad \Sigma \ar[d]^{\flatt_\Sigma}  \\
		\hmonad \hmonad \Sigma \ar[r]_{ \flatt_\Sigma}& \hmonad \Sigma
		}} 	
\end{align*}
 Let  $G \in \hmonad \hmonad \hmonad \Sigma$. 
Define $G_1,G_2 \in \hmonad \Sigma$ to  be the results of applying to $G$ the following functions (to make the comparison easier, we discuss the two hypergraphs in parallel columns for the rest of this proof):
\petrisan{
\begin{align*}
\xymatrix @R=2pc @C=3pc {\hmonad \hmonad \hmonad \Sigma  \ar[r]^{\flatt_{ \hmonad \Sigma}}  & \hmonad \hmonad \Sigma \ar[d]^{\flatt_\Sigma}  \\
		& \hmonad \Sigma
		}	
\end{align*}
}{
\begin{align*}
	\xymatrix @R=2pc @C=3pc {\hmonad \hmonad \hmonad \Sigma   \ar[d]_{\hmonad{ \flatt_\Sigma}} &   \\
		\hmonad \hmonad \Sigma \ar[r]_{ \flatt_\Sigma}& \hmonad \Sigma
		}
\end{align*}
}
Our goal is to show that $G_1$ is the same (isomoprhism type)  as  $G_2$. Unfolding the definition of flattening, we see that the hyperedges in the two sourced hypergraphs are defined as follows.

	\petrisan{Hyperedges in $G_1$ are of the form
	\begin{align*}
((e,f),g)	
\end{align*}
where:
 \begin{itemize}
	\item $e$ is a hyperedge of $G$;
	\item $f$ is a hyperedge in the label of $e$;
	\item $g$ is a hyperedge in the label of $f$;
\end{itemize}
} {Hyperedges in $G_2$ are of the form
\begin{align*}
(e,(f,g)) 	
\end{align*}
where the conditions on $e,f,g$ are the same as for $G_1$.}
From the above description, it follows that the function $\alpha_E$ defined by
\begin{align*}
(e,(f,g)) \mapsto ((e,f),g).
\end{align*}
is  a bijection between the hyperedges from $G_1$ and the hyperedges of $G_2$. This is  because the conditions on $e,f,g$ are the same in both $G_1$ and $G_2$. The function also  preserves  labels and ranks, because  labels and ranks are inherited from $g$ in both $G_1$ and $G_2$. Let us now look at the vertices of the sourced hypergraphs:
\petrisan{
Vertices in $G_1$ are of the form:
\begin{align*}
v \qquad (e,w) \qquad ((e,f),u)	
\end{align*}
where:
\begin{itemize}
	\item $v$ is a vertex of $G$;
	\item $e$ is a hyperedge of $G$;
	\item $w$ is a vertex in the label of $e$ that is not a source in the label of $e$;
	\item $f$ is a hyperedge in the label of $e$;
	\item $u$ is a vertex in the label of $f$ that is not a source in the label of $f$.
\end{itemize}
}{
Vertices in $G_2$ are of form:
\begin{align*}
v \qquad (e,w) \qquad (e,(f,u))	
\end{align*}
where the conditions on $v,e,w,f,u$ are the same as for $G_1$.
}
Define $\alpha_V$ to be the function  from vertices of  $G_1$ to vertices of $G_2$  which is the identity on vertices of the forms $v$ and $(e,w)$ and which is otherwise defined by
\begin{align*}
((e,f),u)) \mapsto (e,(f,u)).
\end{align*}
Again, this function is a bijection, because the conditions on $v,e,w,f,u$ are the same in both $G_1$ and $G_2$.
This function preserves sources, because the source functions in both $G_1$ and $G_2$ are inherited from $G$.  
We are left with showing that incidence is preserved by the functions $\alpha_E$ and $\alpha_V$.  This follows immediately from the following  description of the incidence lists in $G_1$ and $G_2$, which is obtained by unraveling the definitions:
\petrisan{
The incidence list  of a hyperedge 
\begin{align*}
((e,f),g)	
\end{align*}
in $G_1$
is obtained by taking the incidence list of $g$ in the label of $f$, and applying the function which maps a vertex $u$ in the label of $g$ to:
\begin{itemize}
	\item $((e,f),u)$ if $u$ is not a source in the label of  $f$;
	\item $(e,f[i])$ if $u$ is the $i$-th source in the label of $f$, and $f[i]$ is not a source in the label of $e$;
	\item $e[j]$ if $u$ is the $i$-th source in the label of $f$, and $f[i]$ is the $j$-th source in the label of $e$.
\end{itemize}
}{The incidence list  of a hyperedge 
\begin{align*}
(e,(f,g))
\end{align*}
in $G_2$ is obtained by taking the incidence list of $g$ in the label of $f$, and applying the function which maps a vertex $u$ in the label of $g$ to:

\begin{itemize}
	\item $(e,(f,u))$ if $u$ is not a source in the label of  $f$;
	\item $(e,f[i])$ if $u$ is the $i$-th source in the label of $f$, and $f[i]$ is not a source in the label of $e$;
	\item $e[j]$ if $u$ is the $i$-th source in the label of $f$, and $f[i]$ is the $j$-th source in the label of $e$.
\end{itemize}}
\end{proof}


\subsection{Recognisability} In this section we discuss recognisable languages of hypergraphs, i.e.~languages of hypergraphs  that are recognised by $\hmonad$-algebras that are ``finite''. What does ``finite'' mean? This first idea might be to consider algebras that have finitely many elements altogether in the universe. This is a bad idea, for the following reason. For every arity $n \in \Nat$, there is at least one  $n$-ary sourced hypergraph (e.g.~only the sources and no hyperedges), and it must have some value under the product operation in an $\hmonad$-algebra. It follows that an $\hmonad$-algebra is nonempty on every arity, and therefore it cannot have a universe that is finite. A second idea is to consider algebras that are finite on every arity; this is the idea that we use, and it is also the idea that was used by Courcelle when defining \hr-recognisable languages of graphs, see below.
\begin{definition}[Recognisable language of hypergraphs]\label{def:h-recognisable}
	A language $L \subseteq \hmonadzero \Sigma$ is called \emph{recognisable} if it is recognised by a $\hmonad$-algebra that is finite on every arity. 
\end{definition}


In Theorem~\ref{thm:aperio} below, we show that this definition of recognisability coincides with the notion of \hr-recognisability  originally introduced  by Courcelle\myfootcite[Definition 4.29]{Courcelle:2012wq}. The notion of \hr-recognisability is based on a choice of operations on sourced hypergraphs, called the \hr-operations and illustrated in Figure~\ref{fig:hr-operations}, of which the most important is the following \emph{parallel composition} operation. Define the {parallel composition} of two sourced hypergraphs of the same arity to be the sourced hypergraph, also of the same arity,	obtained by taking their disjoint union and then fusing the corresponding sources. 
 We write $\oplus$ for parallel composition; this operation is only defined on pairs of sourced hypergraphs of the same arity. Here is a picture:
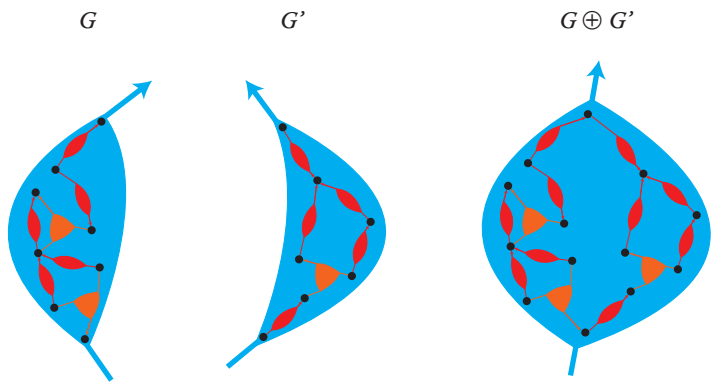
\mypic{45}

To discuss  the relationship  with \hr-recognisability, it will be convenient to view the \hr-operations as a special case of polynomial operations, as defined in  Definition~\ref{def:polynomial}. Consider an algebra of the form $\hmonad \Sigma$, i.e.~the product operation is flattening. Recall that if $X$ is a ranked set of variables, then  a polynomial  operation with variables $X$  in this algebra is a function
\begin{align*}
(\hmonad \Sigma)^X \to \hmonad \Sigma	
\end{align*}
which inputs an arity-preserving valuation of the  variables $X$, and outputs the result of applying this valuation to a  sourced hypergraph $G \in \hmonad (\hmonad \Sigma + X)$ that is fixed for the polynomial operation. The outputs of the polynomial  operation have the same arity as $G$. We are now ready to define \hr-recognisability. 
\begin{definition}[\hr-recognisable language]\myfootcite[Definition 4.29]{Courcelle:2012wq}
	\label{def:hr-recognisable}A language $L \subseteq \hmonadzero \Sigma$ is called \hr-recognisable if there is an  equivalence relation on $\hmonad \Sigma$, which:
	\begin{enumerate}
		\item recognises $L$, i.e.~$L$ is union of equivalence classes; and
		\item has finitely many equivalence classes on every arity; and
		\item is compatible with all of the \hr-operations defined in Figure~\ref{fig:hr-operations}.
	\end{enumerate}
\end{definition}

The following result shows that \hr-recognisability coincides with the notion of recognisability given at the beginning of this section.
 
\begin{theorem}\label{thm:same-as-hr}
	For a language $L \subseteq \hmonadzero{\Sigma}$, the following conditions are equivalent:
	\begin{enumerate}
		\item $L$ is recognisable in the sense of Definition~\ref{def:h-recognisable};
		\item $L$ is recognised by a congruence (in the sense of Definition~\ref{def:congruence}) with finitely many equivalence classes on every arity;
		\item $L$ is recognised by an equivalence relation that is compatible with all \hr-operations and has finitely many equivalence classes on every arity.
	\end{enumerate}
\end{theorem}
\begin{proof}[Sketch]
The equivalence of items 1 and 2 follows from the Homomorphism Theorem.  It remains to prove the equivalence of items 2 and 3.

Define a \emph{linear unary polynomial  operation} to be a polynomial  operation which has  one variable $x$, and which uses this variable exactly once. It is not hard to see that an equivalence relation on $\hmonad \Sigma$ is compatible with all polynomial operations if and only if it is compatible with all linear unary polynomial operations; this is done by replacing each occurrence of each variable one by one (for this argument, it is crucial that the monad is finitary; the argument would not work for infinite hypergraphs).  Therefore, item 2 is equivalent to saying that $L$ is recognised by an equivalence relation that has finitely many equivalence classes on every arity, and is compatible with all linear unary polynomial operations. 

To complete the proof item 2, we show that an equivalence relation is compatible with all linear unary polynomial operations if and only if it is compatible with all \hr-operations. The key observation is that every sourced hypergraph can be obtained by composing the \hr-operations. To prove this observation,  we can start with a sourced hypergraph that has all vertices as sources and no hyperedges, then add  hyperedges one by one using parallel composition, and finally forget some of the sources.  It follows that  every linear unary polynomial  operation can be obtained by composing the \hr-operations, and therefore being compatible with all linear unary polynomial operations is the same as being compatible with all \hr-operations. 
\end{proof}

\begin{figure}
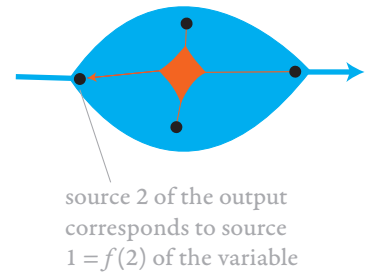

\petrisan{	For each arity $n$, there is an operation  which inputs two $n$-ary sourced hypergraphs, and outputs their  parallel composition.}{\ \vspace{-0.6cm} 
\mypic{33}}

\petrisan{ For each arity $n$, there is an operation which inputs an $n$-ary sourced hypergraph, and adds a new isolated vertex which becomes source $n+1$.}{\ \vspace{-0.6cm}  \mypic{50}}

\petrisan{	There is a constant for the empty hypergraph (arity zero) and a constant for every unit.}{Constants are polynomial operations.}

\petrisan{
For every injective function  
\begin{align*}
  f : \set{1,\ldots,k} \to \set{1,\ldots,n}
\end{align*}
there is an operation which inputs an $n$-ary sourced hypergraph, and returns  a $k$-ary one where the source function of the input is precomposed with $f$.}
{To see that this operation is a polynomial operation, consider the example 
\begin{align*}
f : \set{1,2} \to \set{1,2,3,4} 
\end{align*}
defined by $1 \mapsto 3$ and $2 \mapsto 1$. Then the polynomial corresponding to this operation looks like this:
\mypic{32}}
  \caption{\label{fig:hr-operations}The \hr-operations (left column) and why they are polynomial operations (right column).}
\end{figure}

\subsection{Treewidth}
In this section, we show that a subset of $\hmonadzero \Sigma$, e.g.~a set hypergraphs, has bounded treewidth if and only if it can be generated using finitely many polynomial operations.  
%

\paragraph*{Treewidth.} We begin by defining (a directed hypergraph version of) treewidth. When talking about trees in this paper, we mean finite node labelled unranked trees without a sibling order, as described in Figure~\ref{fig:trees}.

\begin{figure}[hbt]
\mypic{36}
  \caption{Trees \label{fig:trees}}
\end{figure}

\begin{definition}
	[Tree decomposition]\label{def:tree-decomposition} A \emph{tree decomposition} of a hypergraph $G$ is a tree where each node is labelled by a set of vertices  of $G$, which is called the \emph{bag} of the node, such that:
\begin{enumerate}
	\item for every hyperedge of $G$, its incidence list is contained in the bag of some node;
	\item for every vertex $v$ of $G$, the set   
	\begin{align*}
\set{ x : \text{$x$ is a node of the tree $t$ whose bag contains $v$}}
\end{align*}
is nonempty and connected by the child relation on nodes of the tree decomposition.
\end{enumerate}
The width of a tree decomposition is  one plus the maximal size of a bag. The treewidth of a hypergraph is the minimal width of its tree decompositions. 
\end{definition}

Here is a picture of a tree decomposition
\mypic{37} 

The following result shows that bounded treewidth can be expressed in algebraic terms. The proof of the theorem is essentially the observation that a node in a tree decomposition can be viewed as a polynomial  operation (in the algebra $\hmonad \Sigma$), which puts together the hypergraphs generated by its children. 

\begin{theorem}\myfootcite[Proposition 3.7 and Example 4.3(8)]{Courcelle:2012wq}\label{thm:treewidth-grammar}
	Let $\Sigma$ be finite ranked set.  A subset  $L \subseteq \hmonadzero  \Sigma$   has bounded treewidth if and only if there exists a finite set $P$ of polynomial operations (the polynomial operations might use  different, possibly empty, sets of variables)  in $\hmonad \Sigma$ such that	
	\begin{align*}
  L \subseteq \underbrace{\text{least subset of $\hmonad \Sigma$ closed under applying polynomial operations from $P$}}_{\text{subset of $\hmonad \Sigma$ generated by $P$}}.
\end{align*}
\end{theorem}

\begin{proof}
In the proof, we also use tree decompositions also for sourced hypergraphs. A tree decomposition for  a sourced hypergraph is defined the same way as for hypergraphs, except  that we  require   all sources to be  contained in the root bag. In particular the treewidth of a sourced hypergraph is at least its arity minus one.

To prove the left-to-right implication, we use the following result\myfootcite[Theorem 2.83]{Courcelle:2012wq}\label{claim:treewidth-polynomials}, which can be proved by induction on the size of a tree decomposition: a hypergraph has treewidth $< n$ if and only if it can be generated by those operations in Figure~\ref{fig:hr-operations} that use (on input and output) arguments of arity $\le n$. Since these operations are all polynomial operations, and there are finitely many of them, we get the left-to-right implication.

The  right-to-left implication of the theorem follows immediately from the following claim.

\begin{claim}	For every ranked set  $X$ and polynomial operation 
	\begin{align*}
  p : (\hmonad \Sigma)^X \to \hmonad \Sigma
\end{align*}
there  exists some $k \in \Nat$ such that every valuation $\eta \in  (\hmonad \Sigma)^X$ satisfies
\begin{align*}
  \text{treewidth of $p(\eta)$} \quad \le \quad k + \max_{x \in X} \text{treewidth of $\eta(x)$}
\end{align*}
\end{claim}
\begin{proof} The number $k$ is the treewidth of the sourced hypergraph defining $p$.  The number $k$ is necessarily finite, since we are dealing with finite hypergraphs. The result is then proved
by putting together tree decompositions in the obvious way.
\end{proof}

\end{proof}

\section{Monadic second-order logic}
\label{sec:mso}
A classical theme in language theory is to use logic to describe properties of  objects such as  words, trees, graphs, etc. The most prominent logics are first-order logic and  monadic second-order logic (\mso). For more background on this theme, see the survey of Thomas~\myfootcite{Thomas:1990tw}.  The seminal result for this topic that for finite words recognisability (by finite monoids, or equivalently finite automata) is the same thing as definability in \mso; this was shown by B\"uchi, Elgot and Trakhtenbrot. This result was later extended to finite trees, by Thatcher and Wright, and then famously to infinite binary trees  by Rabin (but Rabin's notion of recognisability for infinite trees used nondeterministic automata; finding a suitable algebraic notion of recognisability for infinite trees remains an open question, see ).   For graphs and hypergraphs -- general ones but especially those of bounded treewidth -- the study of the connection between recognisability and definability in \mso was pioneered by Courcelle, and this is the topic that is discussed in this section.
Here is the plan.

\begin{itemize}
\item In Section~\ref{sec:mso-notation}, we introduce  \mso and its counting extension.
	\item In Section~\ref{sec:courcelle-theorem},  we  prove Courcelle's Theorem, which says that for hypergraph languages,  definability  in counting \mso implies recognisability. The proof is phrased  so that it can used for other monads, i.e.~it yields sufficient criteria for a monad to admit ``Courcelle's Theorem''. Examples of such other monads are the monad $\vmonad$ that will be discussed  in Section~\ref{sec:vertex}, as well as all monads for finite and infinite words and trees that are known to the author.  
	
	\item In Section~\ref{sec:aperiodicity} we discuss the expressive power of \mso without counting. We give an algebraic characterisation: a language of bounded treewidth is definable in \mso without counting if and only if it is recognised by an algebra that is finite on every arity and aperiodic (aperiodicity for $\hmonad$-algebras is a condition related to aperiodicity for monoids that was famously used by Sch\"utzenberger\myfootcite{Schutzenberger:1965il} to characterise the star-free languages). 
	\item  In Section~\ref{sec:compute-syntactic} we discuss the algorithmic aspects of checking if an algebra is aperiodic. The difficulty is that  $\hmonad$-algebras have infinitely many sorts and infinitely many operations. We  introduce a notion of computable $\hmonad$-algebra and show that: (a) languages  definable in counting \mso admit computable algebras; (b)  for languages of bounded treewidth, computable algebras can be minimised and tested for conditions such as  aperiodicity.  Combining the results from Sections~\ref{sec:aperiodicity} and~\ref{sec:compute-syntactic}, we get an algorithm which inputs a sentence of counting \mso that defines a language of bounded treewidth, and says whether or not the same language can be defined without counting. Perhaps more importantly, the results from  Sections~\ref{sec:aperiodicity} and~\ref{sec:compute-syntactic} show that  the monad framework is mature enough to get algebraic characterisations of logics (algebraic characterisations are a topic that has been widely studied for words and also for trees);  at least under the assumption of bounded treewidth. 
\end{itemize}

\subsection{\mso over hypergraphs}
\label{sec:mso-notation}
In this section, we introduce notation for \mso.  
 A \emph{vocabulary} is defined to be  a set of \emph{relation names} with associated arities (which are natural numbers, possibly zero) plus a set of \emph{constant names}. Note that we disallow functions that are not constants; this is because we want to have finitely many non-equivalent formulas of first-order logic with given quantifier rank.  A \emph{model} over a vocabulary~$\Sigma$ consists of a set called the  \emph{universe} and an \emph{interpretation}, which maps each $n$-ary relation name in the vocabulary to an $n$-ary relation on the universe and  maps each constant name to an element of the universe. To express properties of models, we use first-order logic (\fo) and monadic second-order logic (\mso). \mso is the extension of first-order logic which allows  quantification  over sets of elements in the universe, see the textbook of Ebbinghaus and Flum~\myfootcite[page 38]{Ebbinghaus:2013kj} for a precise definition. The word \emph{monadic} means that second-order quantification is restricted to sets of elements, as opposed to sets of pairs, triples, etc. 
 

  \mso cannot define  properties like ``the size of the universe  is even'', which are recognised by  algebras that are finite on every arity. To deal with such properties, we consider \emph{counting \mso}, which extends the syntax of \mso by adding, for every $m \in \Nat$ and $k \in \set{0,\ldots,m-1}$  a (second-order) predicate 
\begin{align*}
|X| \equiv k \mod m
\end{align*}
which inputs a set $X$ and returns true if  its size   is  finite and congruent to $k$ modulo~$m$.  The number $m$ is called the \emph{modulus} in the predicate.  We will  consider fragments of counting \mso where the counting predicate above is only allowed for moduli $m$ taken from some set $M$; e.g.~when $M$ is empty we recover the usual \mso.

\paragraph*{Hypergraphs as models.}
To define properties of (sourced) hypergraphs using \mso or counting \mso, we use the following encoding of hypergraphs as models.

\begin{definition}[Sourced hypergraphs as models]\label{def:relational-encoding} Let $\Sigma$ be a finite ranked set. For sourced hypergraph  $G \in \hmonad \Sigma$, define its model  $\structh G$ as follows:
	\begin{enumerate}
		\item The universe is the disjoint union of the vertices and  hyperedges;
		\item For every $a \in \Sigma$ there is a unary relation interpreted as 
		\begin{align*}
  \set{e : \text{$e$ is a hyperedge that has label $a$}}
\end{align*}

		\item For every $i \in \set{1,2,\ldots,\text{maximal arity in }\Sigma}$ there is a binary relation   interpreted as 
		\begin{align*}
  \set{(e,v) : \text{$e$ is a hyperedge whose $i$-th incident vertex is $v$}}
\end{align*}
		\item For every $i \in \set{1,2,\ldots,\text{rank of }G}$ there is a constant for  the $i$-th source.
	\end{enumerate}
\end{definition}
If the ranked set $\Sigma$ is finite, then the vocabulary in the above definition is finite (but depends on the arity of $G$, as used in item 4). In principle, the definition could be applied to infinite alphabets $\Sigma$, in which case the vocabulary would be infinite, but  we do not use infinite variant.  Note that the model defined above has hyperedges in the universe, this issue is discussed in the following example.

\begin{myexample}\label{ex:digraphs-h}   There are two natural encodings of a directed graph $G$ as models, which are denoted as $\lfloor G \rfloor$ and $\lceil G \rceil$ by Courcelle and Engelfriet\myfootcite[Sections 1.3.1 and 1.8]{Courcelle:2012wq}:
\begin{itemize}
	\item[$\lfloor G \rfloor$] the universe is the vertices and  edges are represented by a binary relation;
	\item[$\lceil G \rceil$] the universe is the vertices plus the edges, and  incidence is represented by two  binary relations for source and target of edges.
\end{itemize}
In this section, we are interested in the second type of coding; the first type will be used in Section~\ref{sec:vertex}.  To recover the coding $\lceil G \rceil$, we view a directed graph   as a hypergraph as described in Example~\ref{ex:graphs}, and represent it as a model using Definition~\ref{def:relational-encoding}. 
\end{myexample}

  In Section~\ref{sec:vertex}, we consider a different model encoding $\structv$  (which will apply only to directed graphs, and not hypergraphs). The encoding $\structv$  corresponds to $ \lfloor G \rfloor$ discussed in the above example. The two encodings $\structh$ and $\structv$ lead to different expressive powers of \mso for directed graphs; the difference being quantification over sets of edges.  To avoid confusion between the two encodings, we use the name ``definable in \msoh'' for languages that are defined using the encoding $\structtwo$ from Definition~\ref{def:relational-encoding}; and  the name ``definable in  \msov'' for languages definable using the encoding from Section~\ref{sec:vertex}. The choice of numbers 1 and 2 originates from the graph setting\myfootcite[page 69]{Courcelle:2012wq}, where 2 indicate that sets of edges can be quantified.

\begin{definition}[\msoh]\label{def:msotwo}
A set of hypergraphs $L \subseteq  \hmonadzero \Sigma$ is called definable in (counting) \msoh if there is a sentence $\varphi$ of (counting) \mso over  vocabulary used in Definition~\ref{def:relational-encoding}  that  defines $L$ in the following sense:
	\begin{align*}
G \in L \quad \text{iff} \quad \structtwo G \models \varphi \qquad \text{for every  $G \in \hmonadzero \Sigma$}.	
\end{align*}	
\end{definition}

\subsection{Courcelle's Theorem}
\label{sec:courcelle-theorem}
This section is devoted to Courcelle's Theorem\myfootcite[Section 5.3]{Courcelle:2012wq}. 
\begin{theorem}[Courcelle's Theorem]\label{thm:courcelle}
If a language $L \subseteq \hmonadzero \Sigma$  is definable in counting \msoh, then it  is
	 recognisable (i.e.~recognised by an $\hmonad$-algebra that is finite on every arity).
\end{theorem}
We give a proof of  the theorem below. The point is not to  show that the theorem is true,  which is known, but to write the proof so that it can be resued for other monads.


As discussed at the beginning of Section~\ref{sec:mso}, the  converse implication in Courcelle's  Theorem fails, see Example~\ref{ex:cliques} below.  In the example, it is crucial  to use a language of unbounded treewidth. For  bounded treewidth,  the converse  implication of Courcelle's Theorem is true; we will revisit this issue at the end of this section.

\begin{myexample}\label{ex:cliques}
Consider languages of cliques. For every set $X \subseteq \Nat$ the language
\begin{align*}
  \set{G : \text{$G$ is an undirected graph that is a clique and has size in $X$}},
\end{align*}
is recognisable\myfootcite[Proposition 4.36]{Courcelle:2012wq}, 
assuming the representation of  graphs as hypergraphs described in Example~\ref{ex:graphs}. The basic idea is: for a sourced graph, if there is at least one vertex that is not a source, then the recognising homomorphism needs to produce only one bit of information (is the sourced graph a clique in the language?). For some  choices of $X$, the  language is not  definable in counting \msoh, e.g.~when $X$ is an undecidable set of numbers.
\end{myexample}

\newcommand{\msotwostruct}[1]{\mathsf{mso}#1}
Our proof of Courcelle's Theorem, like the original proof,  uses  the ``composition method'' in logic\footnote{An alternative proof would use the method proposed in~\cite[Section 6]{Anonymous:2015vr}. There, an abstract definition of \mso is given, and then~\cite[Lemma 6.2]{Anonymous:2015vr} is proved showing  that languages defined in this abstract \mso are necessarily recognised by a finite algebras. One reason why we do not use that method is that  proof of Lemma 6.2 contains a  mistake, which   was pointed out by Julian Salamanca: to actually work, the proof of Lemma 6.2 in~\cite{Anonymous:2015vr} requires an additional assumption, namely that there is a distributive law of the monad over powerset. Nevertheless, the method from~\cite{Anonymous:2015vr}, can in fact be used to show the Recognisability Theorem, because: (a) the additional assumption on  a distributive law is satisfied for the monad $\hmonad$; and (b) the abstract notion of \mso used in~\cite{Anonymous:2015vr} is consistent with the notion of \cmso used here. These ideas will be described in an upcoming paper. }. 
The idea behind the composition method is this: if we know the theory of smaller hypergraphs and we know how these hypergraphs are put together to get a bigger hypergraph, then we also  know the theory of the bigger hypergraph.   We begin by describing the composition method for first-order logic, and then lift it to counting \mso. Most of the discussion in this section is about first-order logic and counting \mso in general, without assuming that the logics are evaluated in models representing sourced hypergraphs.



%

\paragraph*{\fo-compatible operations.}
We assume that the reader is familiar with basic notions from model theory, such as quantifier rank and  Ehrenfeucht-Fra\"iss\'e  games for first-order logic\myfootcite{Ebbinghaus:2013kj}. 

\begin{definition}[$r$-equivalence]
For  $r \in \Nat$ and two models $\structa$ and $\structb$, we write
\begin{align*}
  \structa \equiv_r \structb
\end{align*}
if the models have the same vocabulary and  player Duplicator has a winning strategy in the $r$-round  Ehrenfeucht-Fra\"iss\'e  game over these two models.	
\end{definition}
Ehrenfeucht's Theorem\myfootcite[Theorem 1.2.8]{Ebbinghaus:2013kj} says that if the vocabulary is finite, then  $\structa \equiv_r \structb$ holds  if and only if the two models satisfy the same sentences  of first-order logic with quantifier rank at most $r$.  If the vocabulary is finite, then up to logical equivalence there are only finitely many sentences of given quantifier rank, which gives the following fact (here it is important that we have no function symbols of arity $\ge 1$, only constants).


\begin{definition}[Operations compatible with \fo] \label{def:compatible-with-fo} Let $\Sigma_1,\ldots,\Sigma_n,\Sigma$  be vocabularies. An operation 
	\begin{align*}
f :   \text{(models over $\Sigma_1$)} \times \cdots \times \text{(models over $\Sigma_n$)} \to  \text{models over $\Sigma$},
\end{align*}
is called \emph{compatible} with an equivalence relation $\equiv$ on models if 
\begin{align*}
\structa_1 \equiv \structb_1 ,\ldots, \structa_n \equiv \structb_n \quad \text{implies} \quad f(\structa_1,\ldots,\structa_n) \equiv f(\structb_1,\ldots,\structb_n).
\end{align*}
An operation is
 \emph{compatible with \fo} if it is compatible with  $\equiv_r$ for every $r \in \Nat$. 
\end{definition}

We list below four types of  operations on models;  Lemma~\ref{lem:fo-compatible} says that all of them are  compatible with \fo. 
\begin{enumerate}
\item \emph{Products.} Define the product 
 \begin{align*}
	\prod_{i \in I} \structa_i
\end{align*}
of a family of models (possibly over  different vocabularies)
as follows. The universe is the product of the universes. For every $i \in I$ and every relation $R$ in the model $\structa_i$, the product has an $n$-ary relation interpreted as 
\begin{align*} 
\set{(a_1,\ldots,a_n) : R(\pi_i(a_1),\ldots,\pi_i(a_n))}.
\end{align*}
where $\pi_i$ is the projection of the product onto $\structa_i$. For  every family 
\begin{align*}
\set{c_i \text{ is a constant in $\structa_i$}}_{i \in I}	
\end{align*}
the product has a constant that  is interpreted  coordinatewise. 
\item \emph{Disjoint unions.}  Define the disjoint union 
 \begin{align*}
	\coprod_{i \in I} \structa_i
\end{align*}
of a family of models (possibly over  different vocabularies)
as follows. The universe is the disjoint union of the universes. For every  $i \in I$ and every relation $R$ in the model $\structa_i$, the disjoint union has an $n$-ary relation interpreted as 
\begin{align*} 
\set{(a_1,\ldots,a_n) : a_1,\ldots,a_n \in \structa_i \text{ and }R(a_1,\ldots,a_n)}.
\end{align*}
For every $i \in I$ and every constant $c$ in  $\structa_i$, the disjoint union has a corresponding constant. Note that the vocabularies of the product and disjoint union have the same relations, but different constants. 
\item \emph{Quantifier-free universe restriction.} Let $\Sigma$ be a vocabulary, and let $\varphi$ be a quantifier-free formula with one free variable. Using $\varphi$, we  define an operation
\begin{align*}
  \text{models over $\Sigma$} \quad \to \quad   \text{models over $\Sigma$}
\end{align*}
which restricts the universe to elements that satisfy $\varphi$, and restricts all other relations to the new smaller universe.  This operation is partial, because it is undefined if some constant violates $\varphi$.
\item \emph{Quantifier-free interpretation.} Let $\Sigma$ and $\Gamma$ be vocabularies, and $f$ be a function which assigns:
\begin{itemize}
	\item to each $n$-ary relation name in $\Gamma$ a quantifier-free formula over $\Sigma$ with $n$ arguments;
	\item to each constant name in $\Gamma$ a constant name in $\Sigma$.
\end{itemize}
 From $f$ we get a function from models over $\Sigma$ to models over $\Gamma$ as follows: the universe is not changed,  each relation name $R \in \Gamma$  is interpreted according to $f(R)$ and each constant $c$ is interpreted as $f(c)$.  

\end{enumerate}

\begin{lemma}\label{lem:fo-compatible}
	Products, disjoint unions, quantifier-free interpretations, and quantifier-free universe restrictions are compatible with \fo. 
\end{lemma}
\begin{proof}
An application of  Ehrenfeucht-Fra\"iss\'e games. The lemma can be traced back to Mostowski~\myfootcite{Mostowski:1952ew}, see also the discussion of products in Hodges' textbook~\myfootcite[Section 9 and the historical remarks on p. 476]{Hodges:1993kia}.
\end{proof}

\paragraph*{Operations compatible with counting \mso.}  Instead of treating counting \mso as a logic in its own right, it will be convenient  to view it as first-order logic over a suitably defined powerset model. 

\newcommand{\sing}{\mathsf{sing}}
\begin{definition}[Powerset model]\label{def:powerset-model}
For a  model $\structa$, define its  \emph{powerset model} $\powerset \structa$ as follows. The universe of the powerset model is the powerset of the universe of $\structa$, and it is equipped with  the following relations and constants:
\begin{enumerate}
	\item \label{it:powerset-inclusion} A binary relation for set inclusion and a unary relation  for the singleton sets.
	\item \label{it:powerset-singleton} For every $n$-ary relation $R$ in   $\structa$,  a relation of the same name and  arity that is interpreted as 
\begin{align*}
\set{(\set{a_1},\ldots,\set{a_n}) :  \structa \models R(a_1,\ldots,a_n)}.
\end{align*}
	\item \label{it:powerset-filter} For every quantifier-free formula
	 $\varphi(x)$ with one free variable\footnote{The formula $\varphi$ might use non-unary relations, e.g.~it could say $R(x,x)$ or $R(x,c)$ for some constant $c$.} over  $\structa$, a constant $[\varphi]$ interpreted as the elements that satisfy $\varphi$.

%
%
\end{enumerate}
We extend the powerset operation to account for modulo counting as follows. 
For a set $M \subseteq \Nat$, the model $\powerset_M \structa$ is obtained by extending the model $\powerset \structa$ defined above with  the following relations:
\begin{itemize}
	\item [4.] \label{it:powerset-counting} For every $m \in M$ and $k \in \set{0,\ldots,m-1}$, a unary relation selecting sets which are finite and whose size is equal to $k$ modulo $m$.
\end{itemize}
\end{definition}

Some comments about the design choices in the above definition:
\begin{itemize}
\item Some relations and constants in the powerset model can  be defined in terms of others using first-order logic, e.g.~a singleton set is one that has exactly two subsets. However, such definitions are not necessarily quantifier-free, and in the reasoning  below it will be important to use quantifier-free definitions. 
\item  The constants from item~\ref{it:powerset-filter}  will be used to show that  quantifier-free universe restrictions are compatible with counting \mso, see Claim~\ref{claim:commute-cmso-fo} below. Quantifier-free universe restrictions, in turn,  will be needed in the proof of the  Courcelle's Theorem, because the flattening operation of the monad $\hmonad$ removes some vertices, namely the source vertices of the graphs that label hyperedges. 
\item The reason for having  a parameter $M \subseteq \Nat$  in item 4, instead of  $M = \Nat$, is that we want to have finitely many equivalence classes of models for a given quantifier rank. 
\end{itemize}

%

The point of the powerset model is that counting \mso  over a model reduces  to  first-order logic over its powerset, as expressed in the following lemma. To make the correspondence more transparent, we extend the syntax of \mso so that it has a binary (second-order) predicate for set inclusion, a unary (second-order) predicate for testing if a set is a singleton, and for every quantifier-free formula with one free variable $\varphi(x)$, there is a constant $[\varphi$] which is interpreted as in the powerset model. Here is an example of a  quantifier-free sentence in the extended syntax, which says that a binary relation $R$ is reflexive:
\begin{align*}
[\text{true}] \subseteq [R(x,x)].
\end{align*}
Define the quantifier rank of a sentence of counting \mso (in the extended syntax) as for first-order logic, with first- and second-order quantifiers counted the same way. 

%

\newcommand{\cmsomr}{{\sc cmso}$_r^M$\xspace}

\begin{lemma}\label{lem:mso-to-fo}
If models $\structa$ and $\structb$ over the same vocabulary satisfy
\begin{align*}
\powerset_M \structa	 \equiv_r \powerset_M \structb
\end{align*}
if and only if  they satisfy the same sentences of counting \mso (in the extended syntax) that have  quantifier rank $r$ and  use modulo counting for moduli $m \in M$ (we write \cmsomr for the set of such sentences). Furthermore, if the vocabulary is finite and $M$ is finite, then there are finitely many equivalence classes of the above equivalence relation.
%
\end{lemma}
%
\begin{proof}
	By  unraveling the definitions. First-order quantification is replaced by quantification over singleton sets.  For the part of the lemma about finitely many equivalence classes, we observe that although the vocabulary of $\powerset_M \structa$ is technically speaking infinite because of the constants in item~\ref{it:powerset-filter}, there are finitely \emph{different} constants because there are finitely many quantifier-free formulas  up to logical equivalence.
\end{proof}


\begin{definition}[Operations compatible with counting \mso] An operation on  models is called  \emph{compatible with counting \mso} if for every $r \in \Nat$ and $M \subseteq \Nat$ it is compatible with the equivalence relation from Lemma~\ref{lem:mso-to-fo}, i.e.~satisfying the same sentencers of \cmsomr.
\end{definition}

The following lemma shows that all operations from Lemma~\ref{lem:fo-compatible}, except for products, are compatible with \mso. This type of  result was aready known to  Shelah~\myfootcite[Section 2]{Shelah:1975jo}. The  counterexample for products is: a formula of \mso in a product of two finite orders can express that they have the same size.

\begin{lemma}\label{lem:cmso-compatible}
	Disjoint unions, quantifier-free interpretations, and quantifier-free universe restrictions are compatible with counting \mso.
\end{lemma}
\begin{proof}
The key  is  that the  operations in the lemma commute with powersets in a suitable sense, as stated in the following claim. The claim is proved by simply substituting quantifier-free formulas inside other quantifier-free formulas, plus the observation that the powerset operation turns union into product.

\begin{claim}\label{claim:commute-cmso-fo}
Let $M \subseteq \Nat$. Then
\begin{enumerate}
	\item For every quantifier-free universe restriction $f$ 
there is  a quantifier-free universe restriction $g$ such that every model $\structa$ over the input vocabulary of $f$ satisfies
	\begin{align*}
\powerset_M(f(\structa)) \quad \text{is isomorphic to} \quad  g(\powerset_M(\structa)).
\end{align*}
\item Same as item 1, but with both $f$ and $g$ being quanitifer-free interpretations;
\item  For every family of vocabularies $\set{\Sigma_i}_{i \in I}$ there is a quantifier-free interpretation $g$ such that every family $\set{\structa_i \in \text{models over $\Sigma_i$}}_{i \in I}$ satisfes
	\begin{align*}
	\powerset_M (\coprod_{i \in I} \structa_i) \quad \text{is isomorphic to} \quad  g\big( \prod_{i \in I} \powerset_M \structa_i \big).
\end{align*}
\end{enumerate}
\end{claim}
\begin{proof} 
We prove each item separately.
\begin{enumerate}
	\item If the universe restriction $f$ is given by a quantifier-free formula $\varphi$, then the universe restriction $g$  is given by the formula $x \subseteq  [\varphi]$, where the constant $[\varphi]$  comes from item 3 of Definition~\ref{def:powerset-model}.
\item Substituting quantifier-free formulas.
\item 
Consider two models
\begin{align*} \label{eq:coproduct-and-product}
	\powerset_M (\coprod_{i \in I} \structa_i) \qquad\text{and} \qquad    \prod_{i \in I} \powerset_M \structa_i.
\end{align*}
Define the $i$-th component of a set in the disjoint union  to be its intersection with the elements from the $i$-th model. 
The universes of two models are isomorphic, the isomorphism sends a set in the disjoint union to the tuple of its components. We only discuss the constants from item 3 in the definition of the powerset structure, the rest is straightforward. Consider  a constant  in the disjoint union, as in item 3, which is given by  a quantifier-free formula $\varphi$ with one free variable over the model
\begin{align*}
	\coprod_{i \in I} \structa_i.
\end{align*}
For $i \in I$, consider the elements in $\structa_i$ that satisfy $\varphi$. This set can be defined by a quantifier-free formula, call it $\varphi_i$, which uses only the vocabulary of $\structa_i$.  Then the constant $[\varphi]$ in the powerset of the disjoint union  is the same as the tuple of constants $\set{[\varphi_i]}_{i \in I}$ in the product of powersets. 
\end{enumerate}
\end{proof}

Using the above claim, we finish the proof of the lemma. We only treat the case of quantifier-free interpretations, the others are done the same way.  Let then $f$ be a quantifier-free intepretation. Let $\structa$ and $\structb$ be models over the input vocabulary of $f$. To prove that $f$ is compatible with counting \mso, we need to show that for every set $M \subseteq \Nat$ and every quantifier rank $r$, 
\begin{align*}
\powerset_M \structa \equiv_r	\powerset_M \structb \qquad \text{implies} \qquad \powerset_M f(\structa) \equiv_r	\powerset_M f(\structb).
\end{align*}
Apply the claim to $f$ and $M$, yielding some $g$, which is \fo-compatible by Lemma~\ref{lem:fo-compatible}. The conclusion of the  above implication is then proved as follows:
\begin{eqnarray*}
		\powerset_M f(\structa) &\equiv_r& \text{\small{(by the  claim, and because isomorphism refines $\equiv_r$)}}\\
g(\powerset_M \structa)  &\equiv_r& \text{\small{(applying compatibility with \fo  of $g$ to the assumption of the implication)}}\\
g(\powerset_M \structb)		 &\equiv_r& \text{\small{(by  the claim, and because isomorphism refines $\equiv_r$)}}\\
\powerset_M f(\structb).
\end{eqnarray*}
\end{proof}

\paragraph*{Logical decomposition of the monad $\hmonad$.} So far, the discussion  had nothing to do with monads in general, or the  monad $\hmonad$ in particular. The only  part of our reasoning which is specific to the monad $\hmonad$ and the particular encoding from Definition~\ref{def:relational-encoding} is the following lemma.
\begin{lemma}[Compositional encoding]\label{lem:good-encoding} For a ranked set $\Sigma$, define $\structh \Sigma$ to be the image of the set $\hmonad \Sigma$ under the function $\structh$; this image is a viewed as a ranked set. For every ranked set 
$X$ and every  $G \in \hmonad X$ there  is an operation $f$ compatible with counting \mso which makes the following diagram commute:
\begin{align*}
\xymatrix@C=4cm{
(\hmonad \Sigma)^X \ar[r]^{(\structh)^X} \ar[d]_{\sem G} &  (\structh \Sigma)^X \ar[d]^{f}\\
\hmonad \Sigma \ar[r]_{\structh\ } &  \structh \Sigma
}
\end{align*}

\end{lemma}
\begin{proof}
The proof is essentially the observation that the definition of flattening, when working on models representing sourced hypergraphs, can be formalised using disjoint unions, quantifier-free interpretations and quantifier-free universe restrictions; all of which are compatible with counting \mso thanks to Lemma~\ref{lem:cmso-compatible}. 

Let us describe the above observation in more detail.
Let $\eta \in (\hmonad \Sigma)^X$ be a valuation of the variables, and define  $\structa_x$ to be the model that  represents $\eta(x)$. Our goal is to transform  the models $\set{\structa_x}_{x \in X}$ into the model  that representing $\sem G (\eta)$,  using operations compatible with counting \mso. Also, the transformation is only allowed to depend on $G$ and not on the valuation $\eta$. We do this in several steps, all of which are compatible with counting \mso thanks to Lemma~\ref{lem:cmso-compatible}:
\begin{enumerate}
	\item Define $\structb$ to be  model where the universe is the vertices of $G$, and which has a constant for every  vertex. By abuse of notation,   for a hyperedge $e$ in $G$, let us write $\structa_e$ for the model $\structa_x$, where $x \in X$ is the variable that labels $e$. Take the disjoint union of $\structb$ and 
\begin{align*}
 \coprod_{\substack{{e \in \text{hyperedges of $G$}}}} \quad 	 \structa_{e}
\end{align*}
Note that a model $\structa_x$ might be copied several times, since the variable $x$ might be the label of more than one hyperedge of $G$.

\item Using a quantifier-free interpretation applied to the model from the previous item, we recover the relations and constants in the model $\structa$ representing $\sem G(\eta)$. 
The labels of hyperedges are inherited from the summands in the disjoint union. The sources of $\structa$ are taken from $\structb$, and can therefore be defined using the  constants from $\structb$.  Incidence is defined using the following quantifier-free formula in model from previous item as follows. For $v$ and $f$ in the universe of the disjoint union model produced in the previous item, $v$ is the $i$-th vertex in the incidence of a hyperedge $f$   if and only if one of the following conditions holds:
\begin{enumerate}
\item There is some hyperedge  $e$  of $G$ (existence of $e$ is tested using a finite disjunction, and not using existential quantification) such that both $v$ and $f$ are from the model $\structa_e$ and the following quantifier-free formulas are satisfied:
\begin{align*}
\underbrace{\bigwedge_{j \in \set{1,\ldots,\text{rank of $e$}}}\text{$v$ is not the $j$-th source of $\structa_e$}}_{\text{a quantifer-free formula in $\structa_{e}$}} \quad \text{and} \quad   \underbrace{\text{$f[i]=v$  in  $\structa_e$}}_{\text{a quantifer-free formula in $\structa_{e}$}} 
\end{align*}
\item There is some hyperedge  $e$  of $G$ (again, this is a disjunction over hyperedges of $G$) and some $j \in \set{1,\ldots,\text{arity of $e$}}$   such that $v$ is from  $\structb$, and $f$ is from the model $\structa_e$, and the following quantifier-free formulas are satisfied:
\begin{align*}
\underbrace{f[j]=v}_{\text{a quantifer-free formula in $\structb$}} \qquad \text{and} \qquad \underbrace{\text{$e[i]$ is the $j$-th source of $\structa_{e}$}}_{\text{a quantifer-free formula in $\structa_{e}$}}	
\end{align*} 
\end{enumerate}
\item Using a quantifier-free restriction, for each hyperedge $e$ of $G$, remove the vertices of $\structa_e$   that are sources.
\item  Finally, using a quantifier-free interpretation, remove the relations from the vocabulary that are not used in $\structa$.
\end{enumerate}
\end{proof}

\paragraph*{Proof of Courcelle's Theorem.} \label{page:courcelle-theorem} We are now ready to complete the proof of Courcelle's Theorem. Suppose that a language $L \subseteq \hmonadzero{ \Sigma}$ is defined by a sentence of counting \msoh which has quantifier rank $r$ and uses modulo counting  only for moduli $m$ taken from a finite set $M \subseteq \Nat$. Define $\approx$ to be the equivalence relation on  $\hmonad \Sigma$ which identifies sourced hypergraphs  if they have the same rank and (their associated models) satisfy the same sentences of \cmsomr. By definition, the language $L$ is recognised by $\approx$, i.e.~it is a union of equivalence classes.

\begin{lemma}\label{lem:congruence}
The relation $\approx$ is a congruence.
\end{lemma}
\begin{proof}
The statement of this lemma is almost what the Compositional Encoding Lemma says. The only difference is that the Compositional Encoding Lemma talks about polynomial operations which are given by sourced hypergraphs of the form $ G \in \hmonad X$, while in general a polynomial operation is given by a sourced hypergraph of the form
\begin{align*}
G \in \hmonad(\hmonad \Sigma + X).
\end{align*}
We can view $\hmonad \Sigma$ as extra variables, and extend the valuations $\eta_1,\eta_2$ in the statement of the lemma to these new variables via the identity function. The result then follows from the Compositional Encoding Lemma  and the definition of an operation compatible with counting \mso.
\end{proof}

By the Homomorphism Theorem, the ranked set of equivalence classes
$\hmonad \Sigma /_{\approx}$ 
can be equipped with a product operation which turns it into a $\hmonad$-algebra and which turns the quotient function
\begin{align*}
h : \hmonad \Sigma \to \hmonad \Sigma /_{\approx}	 \qquad G \mapsto \text{$\approx$-equivalence class of $G$}
\end{align*}
into a homormophism.   By Lemma~\ref{lem:mso-to-fo}, $\approx$ has finitely many equivalence classes on every rank.  This completes the proof of Courcelle's Theorem\footnote{Courcelle's Theorem is also true for infinite hypergraphs, in the following sense. Consider the variant of the monad $\hmonad$ where the restriction on finiteness is lifted. Then the same statement and the same proof will work for Courcelle's Theorem. The only difference is that infinite products and infinite disjoint unions are used, which requires the following change to the notion of $\equiv_r$. When doing an infinite product or disjoint union, the resulting model has an  infinite vocabulary. For infinite vocabularies,  Ehrenfeucht's Theorem is true, but under the assumption that the syntax of  first-order logic is extended to use infinite conjunctions and disjunctions but finite quantification; this variant of first-order logic is usually denoted by  $L_{\infty,\omega}$, see~\cite[Section 2.2.]{Ebbinghaus:2013kj}. }.

\paragraph*{Courcelle's Conjecture.} 
 As shown in Example~\ref{ex:cliques} at the beginning of this section, the converse of Courcelle's Theorem fails in general, e.g.~there are  languages that are recognisable but  not definable in counting \msoh. Courcelle has conjectured that the implication becomes an equivalence for classes of bounded tree-width. This conjecture, known as Courcelle's Conjecture,  turns out  to be true, in the setting of undirected graphs\myfootcite{Bojanczyk:2016wd}. To compare the versions of Courcelle's conjecture for hypergraphs and undirected graphs, consider  the following transformations between the two settings, which are described in Figure~\ref{fig:graph-to-hypergraph}:
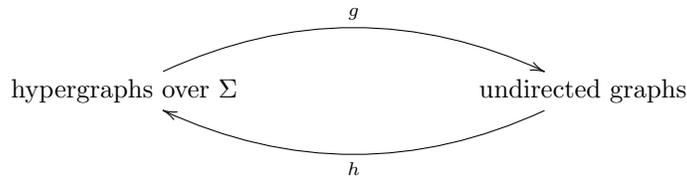
\begin{align*}\xymatrix@C=3cm{\text{hypergraphs over $\Sigma$} \ar@/^2pc/[r]^{g}  &  \ar@/^2pc/[l]^{h} \text{undirected graphs}
}
\end{align*}
(The transformations $g$ and $h$ are not mutual inverses.) One can show without much effort that both $g$ and $h$ preserve and reflect all three notions used in Courcelle's Conjecture: recognisability, bounded treewidth, and definability in counting \msoh. Therefore, since Courcelle's Conjecture is true for undirected graphs, we get it also for hypergraphs:

\newcommand{\edgevoc}{\set{\text{edge}}}
\begin{theorem}\label{thm:courcelle-conjecture}
Let $\Sigma$ be a finite ranked set, and let     $L \subseteq \hmonadzero \Sigma$ be a set of hypergraphs that has bounded  treewidth. Then $L$  is  recognisable if and only if it is definable in counting \msoh.
\end{theorem}

\begin{figure}
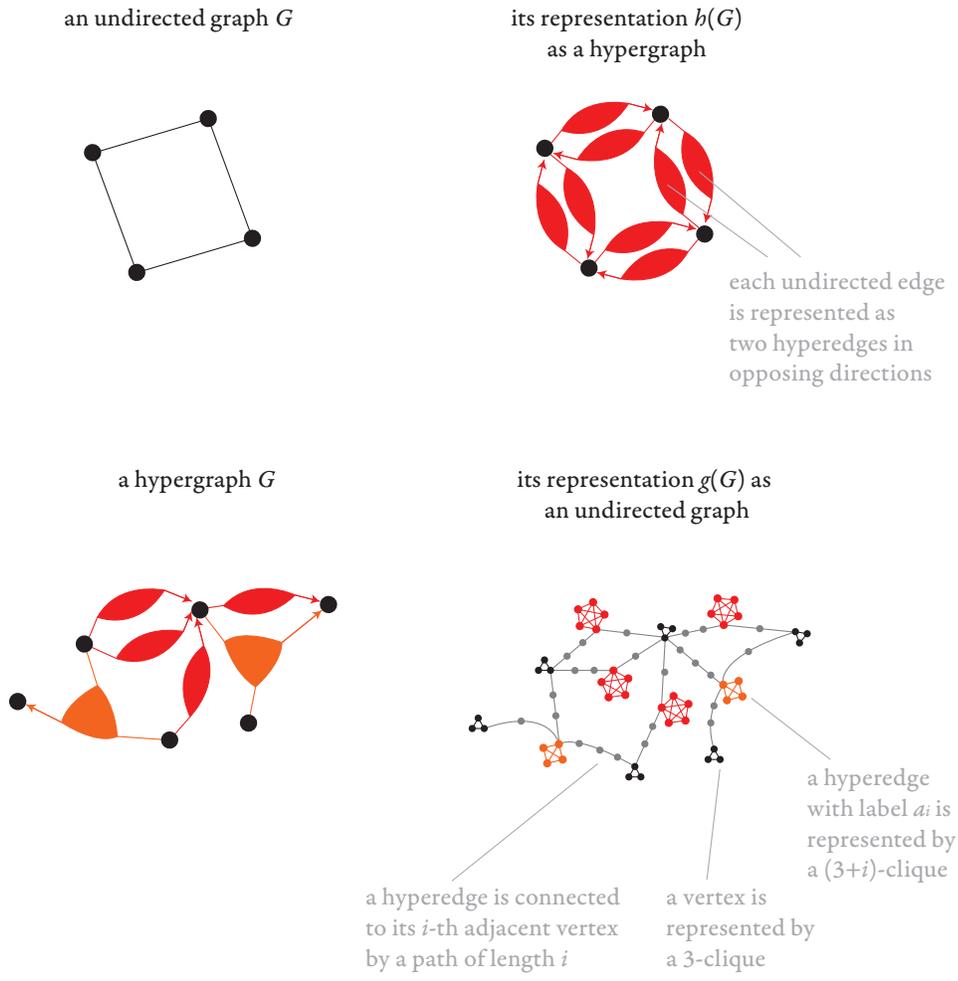

 \mypic{42}
  \mypic{43}
  \caption{\label{fig:graph-to-hypergraph}The transformation $h$, above, represents an undirected graph $G$ as a hypergraph $h(G)$ over a ranked alphabet that has one symbol ``edge'' of rank 2. An alternative definition of $h$ is this:  view an undirected graph as a directed graph with edges directed both ways, and then apply the representation from Example~\ref{ex:digraphs-h}.
  The transformation $g$, below, represents a hypergraph  $G \in \hmonadzero \Sigma$ as an undirected graph $g(G)$. The transformation assumes an  enumeration of the set of labels $\Sigma = \set{a_1,\ldots,a_n}$; the enumeration need not respect the rank information in any way.}
\end{figure}

\subsection{Aperiodicity}
\label{sec:aperiodicity}
In the previous section we have shown that counting \msoh can only define recognisable languages, and furthermore it can define all recognisable languages of bounded treewidth. Counting makes a difference:  \msoh without counting cannot distinguish large independent sets of even size  from ones with odd size. In this section, we discuss this difference in more detail. The main contribution  is Theorem~\ref{thm:aperio}, which says that for  languages of  bounded treewidth, definability in  \msoh without counting is the same as definability by an $\hmonad$-algebra that is finite on every arity and satisfies an algebraic condition called   aperiodicity. Furthermore, since aperiodicity can be checked using an algorithm -- which is not obvious but will be  shown  in Section~\ref{sec:compute-syntactic} -- it follows that there is 
 an algorithm which determines if  a sentence of counting \msoh can be rewritten so that it does not use counting, at least as long as it defines a language of bounded treewidth.   The characterisation using aperiodicity and the accompanying algorithm are, to the author's best knowledge,  new results.

We begin by stating the decidability result, since its statement does not require defining aperiodicity. 
\begin{theorem}\label{thm:decide-mso-not-count}
	The following problem is decidable.
	\begin{itemize}
		\item {\bf Input.} A number $k \in \Nat$ and a sentence $\varphi$  of counting \msoh.
		\item {\bf Question.} Is the following language definable   in \msoh without counting?
		\begin{align*}
\set{ G  : \text{$G$ is an undirected graph of  treewidth $\le k$ that satisfies $\varphi$}}
\end{align*}
	\end{itemize}
	\end{theorem} 

Note that the language in the question from the  theorem is necessarily definable in counting \msoh, because having treewidth $\le k$ is definable in \msoh (e.g.~by testing if the graph contains one of finitely many forbidden minors). 
The above theorem is stated in terms of undirected graphs, but the same result is also true  for hypergraphs. This is because the function $g$ described in Figure~\ref{fig:graph-to-hypergraph}, which codes hypergraphs as undirected graphs, preserves (effectively) bounded treewidth, definability in counting \msoh, and definability in \msoh without counting. It will be convenient, however, to talk about undirected graphs, since this way we can use without modification a result on \mso transductions that was proved for undirected graphs.

The assumption on bounded treewidth is crucial for the decidability result in Theorem~\ref{thm:decide-mso-not-count}.
In general, it is undecidable if a sentence of counting \msoh can be rewritten so that it does not use counting. The intuitive reason is that the unsatisfiable sentence ``false'' does not use counting, and checking satisfiability is undecidable. A formal proof is  longer, but not hard.   

\paragraph*{Aperiodicity.} The proof the decidability result in   Theorem~\ref{thm:decide-mso-not-count} uses an algebraic characterisation, given in Theorem~\ref{thm:aperio} below, which might be  more interesting than the decidability result itself. The algebraic characterisation uses the following notion from semigroup theory,  made famous by Sch\"tzenberger's characterisation of  star-free (equivalently, first-order definable) languages of words~\myfootcite{Schutzenberger:1965il}.
\begin{definition}[Aperiodic semigroup]
	A semigroup $S$ is called \emph{aperiodic} if for every $s \in S$ there exists some $n$ such that
	\begin{align*}
  s^n = s^{n+1}.
\end{align*}
\end{definition}

We will apply the above definition to semigroups generated by  the parallel composition operation $\oplus$ on sourced hypergraphs. The operation $\oplus$ was defined  in algebras of the form  $\hmonad \Sigma$,  but it can be extended to  any $\hmonad$-algebra $\alga$, by defining $a \oplus b$
to be the result of applying the product operation of $\alga$ to
\begin{align*}
	(\unitt a) \oplus (\unitt b).
\end{align*}
It is not hard to see that  $\oplus$ is associative and commutative, and therefore it turns $n$-ary elements in  an $\hmonad$-algebra into an associative and commutative semigroup.


\begin{definition}[Aperiodic $\hmonad$-algebra]\label{def:aperiodic}
	An $\hmonad$-algebra  is called \emph{aperiodic} if for every $n \in \set{0,1,\ldots}$, the semigroup of rank $n$ elements equipped with parallel composition $\oplus$ is  aperiodic.
\end{definition}

The main result of this section is the following theorem, which says that aperiodicity exactly characterises languages definable in \msoh without counting,  under the assumption of bounded treewidth. 

\begin{theorem}\label{thm:aperio} Let $L$ be a language of undirected graphs which has bounded treewidth. Then  the following conditions are equivalent:
\begin{enumerate}
	\item $L$ is definable in \msoh without counting.
	\item $L$ is recognised by an  $\hmonad$-algebra that is aperiodic and  finite on every arity.
\end{enumerate}
\end{theorem}

In the next section, we show  that condition 2 above is decidable, assuming that the language is given by a sentence of counting \msoh together with a bound on its treewidth.  Putting together these results, we get   Theorem~\ref{thm:decide-mso-not-count}.

The rest of Section~\ref{sec:aperiodicity} is devoted to proving Theorem~\ref{thm:aperio}. 

\paragraph*{Definability in \msoh implies aperiodicity.}
 We begin  by proving the  easier top-down implication in   Theorem~\ref{thm:aperio}. This implication does not need the assumption that $L$ has bounded treewidth. 
\begin{lemma}\label{lem:}
If a  language of undirected graphs is   definable in \msoh, then it  is recognised by an $\hmonad$-algebra that is  aperiodic and  finite on every arity.
\end{lemma}
\begin{proof} Consider a language of undirected graphs that is defined by a sentence of  \msoh of quantifier rank $r$. Define $\approx_r$ to be the equivalence relation which identifies two sourced hypergraphs if (their associated models) satisfy the same rank $r$ sentences of  \mso without counting.   By the proof of Courcelle's Theorem, $\approx_r$ is a congruence, and leads to a quotient algebra that is finite on every arity.  To prove that the quotient algebra is aperiodic, and thus complete the proof of the lemma, we use the following claim, which is shown by induction on $r$, using a simple Ehrenfeucht-Fra\"iss\'e argument.
\begin{claim} For every   finite ranked set $\Sigma$, every $G \in \hmonad \Sigma$, and every quantifier rank $r$, there is some $k$ such that
\begin{align*}
\overbrace{G \oplus \cdots \oplus G}^{\text{$k$ times}} \qquad \approx_r   \qquad \overbrace{G \oplus \cdots \oplus G}^{\text{$k+1$ times}}.
\end{align*}
\end{claim}
\end{proof}

\paragraph*{Aperiodicity implies definability in \msoh.} The rest of Section~\ref{sec:aperiodicity} is devoted to proving  the top-down implication in Theorem~\ref{thm:aperio}: if a language of undirected graphs  has bounded treewidth and is recognised by an $\hmonad$-algebra that is aperiodic and finite on every arity,  then it is definable in \msoh without counting.  Like the proof of Courcelle's Conjecture\myfootcite{Bojanczyk:2016wd}, our proof uses \mso transductions that compute tree decompositions.  In the proof, we crucially use the  assumption on  bounded treewidth, but it not clear if the assumption is really needed.

We begin by discussing \mso transductions. For ranked sets $\Sigma$ and $\Gamma$, viewed as vocabularies that contain only relation names, define a \emph{transduction with input vocabulary $\Sigma$ and output vocabulary $\Gamma$}  to be any set of pairs 
\begin{align*}
  (\structa,\structb) \qquad \text{where $\structa,\structb$ are models over vocabularies $\Sigma,\Gamma$, respectively}
\end{align*}
which is closed under isomorphism, i.e.~replacing either the first or second coordinate by an isomorphic model does not affect memebership in the set. An \mso transduction is a special case of a transduction, which is definable using formulas of \mso. Since we use \mso transductions as a black box, we do not give the definition\footnote{For a definition of \mso transductions, see Definition 2.3 in \cite{Bojanczyk:2016wd}}, and we only state the following property that will be used below:
\begin{lemma}[Backwards Translation]\myfootcite[Theorem 7.10]{Courcelle:2012wq} \label{lem:backwards} 
For every \mso transduction $\Tt$  and every sentence  $\varphi$ of \mso  over its output vocabulary,  one can compute a sentence of \mso over the input vocabulary that defines the set:
 \begin{align*}
\set{\structa : \text{there is some $\structb$ such that $(\structa,\structb) \in \Tt$ and $\structb \models \varphi$}}
\end{align*}
\end{lemma}
The lemma above is also true for counting \mso, but we will use the non-counting version described above.
 We will be interested in \mso transductions that input undirected graphs  and output labelled trees that represent their tree decompositions. By trees we mean node labelled, unranked and unordered trees, as described in  Figure~\ref{fig:trees}. We view trees labelled by  a finite (not ranked) set as a special case of  hypergraphs,
see Figure~\ref{fig:trees-as-hypergraphs}.
Using this representation, it makes sense to talk about a tree language being recognisable or definable in (counting) \msoh. These notions coincide with the usual notions of recognisability and definability for trees. 
\begin{figure}
\mypic{44}
  \caption{Trees as hypergraphs. \label{fig:trees-as-hypergraphs} }
\end{figure}

\paragraph*{Algebraic tree decompositions.}  We now discuss how tree decompositions are represented as trees labelled by operations in the algebra $\hmonad \edgevoc$.  The idea is the same as for Theorem~\ref{thm:treewidth-grammar}, i.e.~a tree decompositions is viewed as a tree labelled by polynomial operations.   An important difference with respect to Theorem~\ref{thm:treewidth-grammar}  is that we need to use trees of unbounded branching. Intuitively speaking, the reason for using high branching is that creating a tree decomposition with bounded branching, say binary branching, would require finding an ordering on the input graph, which cannot be done in \mso. For example, for a graph which  is an independent set like this:
\mypic{38}
we will need to consider tree decompositions that look like  this: 
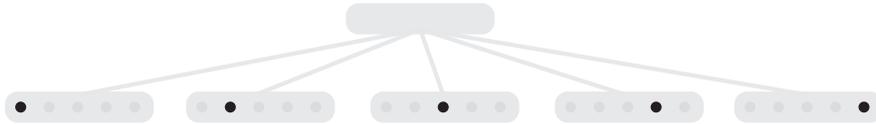
\mypic{39}  
In  the proof of Theorem~\ref{thm:treewidth-grammar} we did not have the issue described above, because we only cared about the existence of a tree decomposition, and it was not important that this tree decomposition would have to be produced using \mso. The above discussion motivates the  following definition, especially item 3:

\begin{definition}[Algebraic tree decomposition] An \emph{algebraic tree decomposition} is a tree (of possibly unbounded branching), where:
\begin{enumerate}
	\item nodes of degree 0 (leaves) are labelled elements  of $\hmonad \edgevoc$;
	\item nodes of degree 1 are labelled by polynomial operations in  $\hmonad \edgevoc$ with 1 argument;
	\item nodes of degree $\ge 2$ have labels of the form $\oplus_n$ for some $n \in \Nat$.
\end{enumerate}
\end{definition}

%
%
 The \emph{value}   of an algebraic tree decomposition is an element of $\hmonad \edgevoc$ that is defined in the natural way, by induction on its  size: 
 if  an algebraic tree decomposition has root label $f$ and $n$ children, then its value is defined to be
 \begin{align*}
f(\text{value of child $1$}, \ldots, \text{value of child $n$}).
\end{align*}
The value is undefined if either: (a) $n=1$ but the value of the unique child is undefined or has a different arity than the argument of the polynomial  operation $f$; or (b) $n \ge 2$, the label $f$ is $\oplus_k$,  but some child has undefined value or a value of arity different than $k$. Note that the trees in algebraic tree decompositions are unordered (on siblings).  However, since  the operation $\oplus_n$ is commutative,  the value described above does not depend on the ordering on the children.

We are ultimately interested in algebraic tree decompositions that produce elements of arity 0, nevertheless subtrees might produce values of other arities.

 We are now ready to state the main technical result from the proof of Courcelle's Conjecture, which says that an \mso transduction can compute tree decompositions for undirected graphs of bounded treewidth.

\begin{theorem}\label{thm:mso-transduction-treedecomp}\myfootcite[Theorem 2.4 and the proof of Lemma 2.11]{Bojanczyk:2016wd}
	For every $k \in \Nat$  one can compute  a finite set 
	\begin{align*}
  \Delta \subseteq \text{polynomial operations in $\hmonad \edgevoc$ with $\le 1$ arguments} \cup \set{ \oplus_n : n \in \Nat}
\end{align*}
and an \mso transduction
	\begin{align*}
 \Tt \subseteq \set{(G,t) : 
 \begin{cases}
 \text{$G$ is an undirected graph}\\
 \text{ 	$t$ is an algebraic tree decomposition with labels from $\Delta$ and value $G$}
 \end{cases}
}
\end{align*}
such that for every graph $G$ with treewidth at most $k$, the set $\mathcal \Tt(G)$ is nonempty.
\end{theorem}

Using the above theorem, we can complete the  proof of Theorem~\ref{thm:aperio}. 
We have already shown the top-down  implication (languages definable in \mso are recognised by algebras that are aperiodic and finite on every arity). It remains therefore to prove the bottom-up implication, under the additional assumption that the language has bounded treewidth. Let then $L$ be a language of undirected graphs, viewed as a subset of $\hmonadzero \edgevoc$ according to the representation in Figure~\ref{fig:graph-to-hypergraph}. Assume that $L$ has  bounded treewidth, say treewidth at most $k$, and is recognised by a homomorphism \begin{align*}
  h : \hmonad \edgevoc \to \alga
\end{align*}
such that $\alga$ is aperiodic and finite on every arity. Our goal is to show that $L$ is definable in \msoh without counting.  Apply Theorem~\ref{thm:mso-transduction-treedecomp} to $k$ yielding a finite set  $\Delta$ and an \mso transduction $\Tt$. 
The rest of the  proof below is described diagramatically in  Figure~\ref{fig:aperiodicity}. 

\begin{figure}[hbt]
\begin{align*}
\xymatrix @C=1.2cm @R=2cm
{
&\text{(undirected graphs of treewidth $\le k$)} \subseteq \hmonadzero \edgevoc
\ar@{.>}[d]^\Tt \ar@{^{(}->}[ddl]_{\text{id}} \ar@{.>}[ddr]^{h \circ \text{value} \circ \Tt}
\\
& \text{trees labelled by $\Delta$} \ar[dl]^{\text{value}} \ar[dr]_{h \circ \text{value}}\\
\hmonadzero \edgevoc \ar[rr]_{ h} & &  A
}	
\end{align*}
  \caption{\label{fig:aperiodicity}The right and bottom small triangles in the diagram commute by definition, while the  left small  triangle commutes by Theorem~\ref{thm:mso-transduction-treedecomp}. It follows that the big  triangle commutes (in particular, the right dotted line in the diagram can be made solid, because it is a composition of two functions).  The function $h \circ \text{value}$ is definable in \msoh thanks to Lemma~\ref{lem:boneva}, and therefore  $h \circ \text{value} \circ \Tt$ is definable in \msoh thanks to the Backwards Translation Lemma. Because the big triangle commutes, it follows that $h$ is also definable in \msoh.
}
\end{figure}
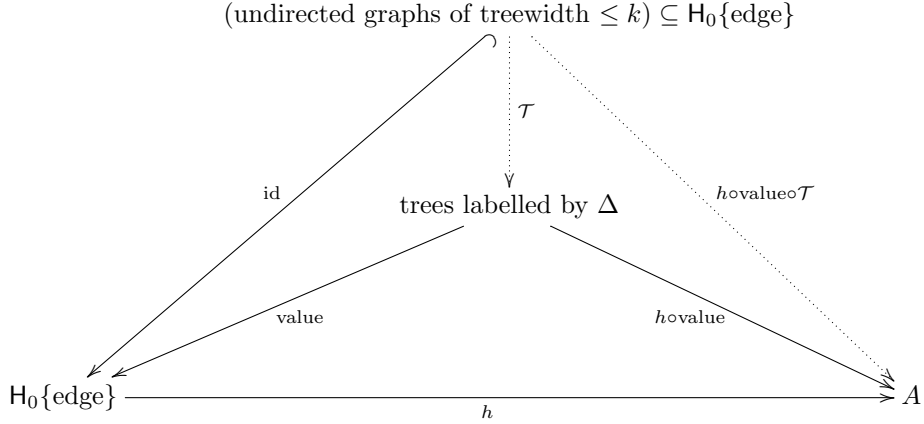

   \begin{lemma}\footnote{The proof of this lemma is the same as in \cite[Theorem 6]{Boneva:2005ks}}
   \label{lem:boneva}
For every  $a \in \alga$ of arity zero, the language
\begin{align*}
  (h \circ \text{value})^{-1}(a) \subseteq \text{trees labelled by $\Delta$}
\end{align*}
is definable in  \msoh without counting
\end{lemma}
\begin{proof}
The main observation is this. If $S$ is a  commutative semigroup, then by definition of commutativity the product of a list $s_1,\ldots,s_n$ depends only on the number of ocurences of each $s \in S$ in that list. If $S$ is furthermore aperiodic, then there is some $n_0 \in \Nat$ such that the number of ocurrnces need only be counted up to threshold $n_0$. We will apply this observation to the aperiodic commutative semigroups of $n$-ary elements in $\alga$ equipped with $\oplus_n$.

To define the language in the statement of the lemma, we do the following. Using existential set quantification,  for  every node $x$ in the input tree  we guess
\begin{align*}
a_x =  (h \circ \text{value})(\text{subtree of $x$}) \in \alga,
\end{align*}
such that the root is labelled by $a$.
In order for this guessing to be possible, we need to justify that  there are finitely many candidates for $a_x$. This is true because the arity of $a_x$ is determined by the (output type) of the label of the node $x$, there are finitely many labels in $\Delta$, and the algebra $\alga$ is finite on every arity.
 Once the values $\set{a_x}_x$ have been guessed, it remains to check if they are consistent with each other.  For leaves, there is nothing to do, and for a node $x$ of degree one it suffices to check if the guessed elements for $x$ and its child are consistent the operation labelling $x$. For nodes labelled by $\oplus_n$, we use the observation from the beginning of the proof, which implies that checking consistency only requires counting the labels in the children up to a threshold (without any modulo counting), something that can be done in first-order logic\footnote{\label{foot:use-small-ranks} Note that in this lemma, we only use the aperidiodicity assumption for $\oplus_n \in \Delta$.
}.
\end{proof}

   By the assumption that $L$ is recognised by $h$, there is a finite set $F$ of nullary elements in $\alga$ such that $L = h^{-1}(F)$. Define $K$ to be the union of the tree languages described in Lemma~\ref{lem:boneva}, ranging over elements $a \in F$. The tree language $K$ is definable in \msoh without counting, as a finite union  of tree languages with this property.    By Theorem~\ref{thm:mso-transduction-treedecomp}, 
   \begin{align*}
  L = \Tt^{-1}(K),
\end{align*}
and therefore $L$ is definable in \msoh without counting thanks to Lemma~\ref{lem:backwards}.
\subsection{Computing a syntactic algebra}
\label{sec:compute-syntactic}
In this section, we show that the aperiodicity condition discussed in Theorem~\ref{thm:aperio}  can be checked using an algorithm. (This algorithm will have the usual non-elementary running time that appears when dealing with \mso.) Checking aperiodicity is only a pretext for discussing a more general question: how can we  effectively represent $\hmonad$-algebras, and do operations on them? For finite monoids, groups, rings, etc., representation is straightforward: one simply provides a finite multiplication table for each of the  finitely many operations. For $\hmonad$-algebras, this straightforward representation does not work, because (a) the universe is necessarily infinite due to the  infinitely many arities; and (b) the product operation has an infinite domain. To deal with these issues, we use the following representation for $\hmonad$-algebras.


\begin{definition}[Computable algebra]\label{def:computable-algebra}
	An $\hmonad$-algebra  is called \emph{computable} if elements of its universe $A$ can be represented in a finite way (say, as strings over a finite alphabet) so that there are algorithms which do the following:
	\begin{enumerate}
		\item  given an arity $n \in \Nat$ and a number $i \in \set{1,2,\ldots}$,  compute the $i$-th $n$-ary element (or output that there are $<i$ $n$-ary elements);
		\item given  an element of $\hmonad A$,  compute its product.
	\end{enumerate}
\end{definition}

An example of a computable $\hmonad$-algebra is the free algebra $\hmonad \Sigma$ for every finite ranked set $\Sigma$. This example is infinite on every arity, but our main focus is on computable algebras that are finite on every arity. Theorem~\ref{thm:computable-courcelle} below says that all languages definable in counting \msoh are recognised by algebras that are computable and finite on every arity; this result  can be viewed as a computable version of Courcelle's Theorem. To prove Theorem~\ref{thm:computable-courcelle}, we use an algebraic approach, which can be viewed as an alternative to the logical proof of  Courcelle's Theorem that we presented in Section~\ref{sec:courcelle-theorem}. For the sake of brevity, we only give a very rough sketch of the algebraic proof;  the computable version of Courcelle's Theorem is not used elsewhere in the paper.

\begin{theorem}[Computable Courcelle's Theorem]\label{thm:computable-courcelle}
If a language of hypergraphs is definable in counting \msoh, then it is recognised by a computable algebra that is finite on every arity.
\end{theorem}
\begin{proof}[Rough sketch]
The construction is going to be uniform in the following sense: if we are given the sentence of counting \mso, then we can return the algorithms that are used in the representation of a computable algebra. 

The key observation, which dates back to B\"uchi, is that quantification in logic can be simulated on the algebra side by a powerset construction, defined as follows.
If $\alga$ is an $\hmonad$-algebra with universe $A$, then define its  \emph{powerset algebra}\footnote{Powerset algebras for sourced hypergraphs appear in Section 3.1.1 of the book by Courcelle and Engelfriet. The proof of Theorem~\ref{thm:computable-courcelle}   can also be seen as a special case of  Lemma 6.2 in: 
\myfootcite{Anonymous:2015vr}. It is worth pointing out that the proof is rather delicate, as witnessed by the following error in the aformentioned paper. In the proof of Lemma 6.2, it is claimed that this powerset algebra is well defined for every monad over sorted sets (and other categories). That claim is false, which I learned from  Julian Salamanca, together with the following counterexample: the  free group monad in the category of sets. Nevertheless, the powerset operation does work for the particular monad $\hmonad$. Together with Julian Salamanca and Bartek Klin, we are planning a paper that discusses the conditions on the monad which ensure that the powerset operation is well defined. 
} to be the algebra where the universe is the powerset $\powerset A$ (i.e.~$n$-ary elements are sets of $n$-ary  elements in $A$) and the product operation is defined by  
 \begin{align*}
G \in \hmonad (\powerset A) \qquad \mapsto \qquad \set{\text{product of }G' : \text{$G' \in_\hmonad G$ }}.
\end{align*}
In the above,  $G' \in_\hmonad G$ means that $G'$ is a sourced hypergraph that has the same vertices, hyperedges and sources as $G$, and the labeling is such that for each hyperedge $e$, the label of $e$ in $G'$ belongs to the set that labels $e$ in $G$.  One can show that this product  operation  satisfies the axioms required in an $\hmonad$-algebra. Furthermore, if $\alga$ is computable, then so is its   powerset.

The theorem follows immediately from the following claim. Another source of inspiration for the claim can be Shelah's ``theory'' function\myfootcite[Definition 2.2]{Shelah:1975jo}, which is designed in such a way that its product operation is computable. 
\begin{claim}\label{claim:computable-courcelle} For every  finite $M \subseteq \Nat$, every finite ranked set $\Sigma$ and every  $r \in \Nat$ there is a homomorphism 
\begin{align*}
  h : \hmonad \Sigma \to \alga
\end{align*}
such that $\alga$ is computable\footnote{Here is a delicate point: the homomorphism $h$ constructed in the claim is not surjective, and in fact computing its image is an undecidable problem. This is because if we could compute the image, then we could decide satisfiability for \msoh on hypergraphs, which is an undecidable problem.}, finite on every arity, and recognises all languages $L \subseteq \hmonad \Sigma$ that are definable in \cmsomr.  	
\end{claim}

The claim is proved by induction on $r$, i.e.~by induction on the quantifier rank.  For the induction base,  one goes  through all the predicates, such as modulo counting or adjacency, applied to the constants available in  \cmsomr (recall that we allow constants for sets defined by quantifier-frees formulas). For the induction step, one applies the powerset construction to the algebra from the induction assumption.
\end{proof}

\paragraph*{The syntactic homomorphism.} The main goal of this section is to complete the proof of Theorem~\ref{thm:decide-mso-not-count} about deciding if a language  of bounded treewidth  can be defined in \msoh without counting. By  Theorem~\ref{thm:aperio}, the question boils down to checking if a language is recognised by some algebra that is aperiodic and finite on every arity. To this end, the computable version of Courcelle's Theorem is not useful, because the algebra produced in the theorem might be not aperiodic, even if the language is recognised by some other algebra that is aperiodic. In fact, a closer inspection of the proof of  Theorem~\ref{thm:computable-courcelle} reveals that the algebras it produces are  \emph{never}  aperiodic in  nontrivial cases, i.e.~when the language is given by a sentence that uses  modulo counting.

To avoid the issues discussed above, we will compute the syntactic algebra for a language, see Definition~\ref{def:synt-algebra} below.  The syntactic algebra is minimal among those recognising the language; a corollary is that  if any recognising algebra is aperiodic, then the syntactic one is, too.

\begin{definition}[Syntactic homomorphism and algebra] \label{def:synt-algebra} A homomorphism
\begin{align*}
  h : \hmonad \Sigma \to \alga
\end{align*}
is called  \emph{syntactic} for  a language $L \subseteq \hmonadzero \Sigma$ if it recognises $L$, and for every surjective homomorphism
\begin{align*}
  g :\hmonad \Sigma \to \algb
\end{align*}
that also recognises $L$, there is  a surjective homomorphism $f$ which makes the following diagram commute
\begin{align*}
  \xymatrix{ \hmonad \Sigma \ar[r]^h \ar[dr]_g  & \alga \\ &  \algb \ar[u]_f}
\end{align*}
The algebra $\alga$ is called a \emph{syntactic algebra} of $L$. 
\end{definition}
The syntactic algebra is unique up to isomorphism, and syntactic homomorphisms are also unique up to isomorphisms on their target algebras.  
 This is why below  we will talk about \emph{the} syntactic homomorphism/algebra of a language.  The following result gives a characterisation of the syntactic homomorphism in terms of parallel composition equivalence.
 \begin{lemma}\myfootcite[Theorem 4.34]{Courcelle:2012wq}\label{lem:hr-synt-par}
  Let  $ L \subseteq \hmonadzero \Sigma$. Two sourced hypergraphs $G,G' \in \hmonad \Sigma$ have the same value under the syntactic homomorphism if and only if they have the same arity,  say $n$,  and
  	\begin{align*}
G \oplus H \in L \quad \text{iff} \quad 	G' \oplus H \in L \qquad \text{for every } H \in \hmonadn \Sigma n
\end{align*}
where membership in $L$ is tested after forgetting the sources
\end{lemma}
\begin{proof}  Apart from parallel composition equivalence, we use  two other equivalence relations: Myhill-Nerode equivalence and its linear variant.  Define  \emph{Myhill-Nerode equivalence} to be the equivalence relation which identifies $G, G' \in \hmonad \Sigma$ if they have the same arity, say $n$, and 
	\begin{align}
f(G) \in L \quad \mbox{iff} \quad f(G') \in L \qquad \text{for every polynomial operation }f : \hmonadn \Sigma n \to \hmonadzero \Sigma.
\end{align}
The  linear variant is defined the same way, except that $f$ ranges only over 
 \emph{linear} polynomial operations, i.e.~ones that are induced by some sourced hypergraph where the  variable appears exactly once. 
 For the same reasons as discussed in the proof of Theorem~\ref{thm:same-as-hr}, all
 three equivalences (parallel composition, and both variants of Myhill-Nerode equivalence) are the same.  Nevertheless, it is the variant with parallel composition that is going to be most useful later on, and therefore this is the one that we use in the statement of the lemma.

To complete the proof, we use the following general result\myfootcite[Lemma 3.6]{Anonymous:2015vr}: in every finitary monad over a category of sorted sets, a language is regognised by an algebra that is finite on every sort if and only if Myhill-Nerode equivalence has finitely many equivalence classes on every sort.  The monad $\hmonad$ satisfies the assumption of the general result, because of the restriction to sourced hypergraphs with finitely many  hyperedges. Since Myhill-Nerode equivalence is the same as parallel composition equivalence, the result follows. 
\end{proof}

\begin{theorem}\label{thm:compute-syntactic-algebra}
If a hypergraph language is definable in counting \mso and has  bounded treewidth, then its  syntactic algebra is computable. 
\end{theorem}
The assumption on bounded treewidth is needed, as will be shown in  Example~\ref{ex:uncomputable-syntactic}.

\begin{proof}
The construction is going to be uniform in the following sense: if we are given the sentence of counting \mso, and a bound on the treewidth, then we can return the algorithms that are used in the representation of a computable algebra. 

Let  $L \subseteq \hmonadzero \Sigma$ be definable in counting \msoh. Let $\sim$ be the equivalence relation in Lemma~\ref{lem:hr-synt-par}.
By the lemma and the Homomorphism Theorem,  we need to show that $\hmonad \Sigma/_\sim$ is computable. 
Suppose that $L$ is defined by a sentence of counting \mso which has quantifier rank $r$ and uses moduli from a finite set $M \subseteq \Nat$. Define $\approx$ to be the equivalence relation from the proof of Courcelle's Theorem in Section~\ref{sec:courcelle-theorem}, i.e.~two sourced hypergraphs are equivalent under $\approx$ if they have the same arity and satisfy the same sentences of \cmsomr. 
As shown in the proof of Courcelle's Theorem,  $\approx$ is a congruence, has finitely many equivalence classes on every arity, and recognises $L$. The following claim is implicit in the proof of Lemma~\ref{lem:mso-to-fo}.

\begin{claim}\label{claim:finite-cmso}
	For every $n \in \Nat$, one can compute a finite set $\Delta_n$ of sentences of counting \mso such that $n$-ary sourced hypergraphs are  equivalent under $\approx$ if and only if they satisfy the same sentences from $\Delta_n$.
\end{claim}

A corollary of the above  claim is that one can test if two $n$-ary sourced hypergraphs are equivalent under $\approx$ by simply evaluating the sentences from $\Delta_n$. In particular, $\approx$ is decidable\footnote{By decidable we mean that it is decidable if two sourced hypergraphs are equivalent under $\approx$. In contrast, one cannot compute the number of equivalence classes of $\approx$ on a given arity. The essential reason is that  satisfiability of counting \mso is undecidable for hypergraphs of unbounded treewidth. Using this undecidability result, one can easily show that there  is no way of computing the number of $\approx$-equivalence classes on arity $n$, assuming that the parameters $M,r$ and $n$ are all part of the input. If we assume that $M$ and $r$ are fixed and the only input parameter is $n$, then it is still the case that one cannot compute the number of equivalence classes, but the proof requires a bit more care, see  Example~\ref{ex:uncomputable-syntactic}.

In particular $\hmonad \Sigma / \approx$ is not a computable algebra, for certain choices of $M$ and $r$; actually whenever $r$ is large enough.
Contrast this with the construction in the computable version of Courcelle's Theorem from Theorem~\ref{thm:computable-courcelle}, which produced a computable algebra. The difference is that the homomorphism into $\hmonad \Sigma /_\approx$ is surjective, unlike the homomorphism used in Theorem~\ref{thm:computable-courcelle}, and this surjectivity comes at the cost of non-computability.}. 

Let $k$ be a bound on the treewidth of $L$. 
 Define $S \subseteq \hmonad \Sigma$ to  be those sourced hypergraphs which have treewidth $\le k$ after forgetting the sources\footnote{Contrast this with   the notion of treewidth for sourced hypergraphs used in the proof of Theorem~\ref{thm:treewidth-grammar}, where the sources were required to be in the root bag. Forgetting the sources is used in the ideal condition discussed in the proof of Claim~\ref{claim:two-other-parallel}
.}, and define $S_n$ to be the $n$-ary elements in $S$. 

\begin{claim}
Given $n \in \Nat$, one can compute a finite set $\Hh_n \subseteq S_n$ which represents all $\approx$-equivalence classes of $S_n$.
\end{claim}
\begin{proof}
For every $\Gamma \subseteq \Delta_n$, we test if there is some sourced hypergraph in $S_n$ which satisfies all sentences from $\Gamma$ and violates all sentences from $\Delta_n -\Gamma$.  This test can be done effectively, because satisfiability of counting \mso is decidable on hypergraphs of bounded treewidth\myfootcite[Theorem 5.80]{Courcelle:2012wq}.  Once we know that there exists sourced hypergraph which passes the test, we can find an example, e.g.~through exhaustive enumeration.
\end{proof}

Using the set $\Hh_n$ from the above claim, we  present an alternative characterisation of  $\sim$, see~\eqref{eq:smalltw-upto-h-envs} in the  claim below, which will allow us to decide $\sim$.

\begin{claim}\label{claim:two-other-parallel} The three conditions below are equivalent for every  $G,G' \in \hmonadn \Sigma n$:
\begin{eqnarray}
\label{eq:all-h-envs}G \oplus H \in L \quad \text{iff} \quad G' \oplus H \in L  &\qquad \text{for every }&H \in \hmonadn \Sigma n\\
\label{eq:smalltw-h-envs} G \oplus H \in L \quad \text{iff} \quad G' \oplus H \in L & \qquad \text{for every }&H \in S_n\\
\label{eq:smalltw-upto-h-envs} G \oplus H \in L \quad \text{iff} \quad G' \oplus H \in L & \qquad \text{for every }&H \in \Hh_n.	
\end{eqnarray}
\end{claim}
\begin{proof}
In this proof, we use the name \emph{environment} for the sourced hypergraph $H$ in any of the conditions~\eqref{eq:all-h-envs}, \eqref{eq:smalltw-h-envs} and~\eqref{eq:smalltw-upto-h-envs}. 
Clearly we have the implications~\eqref{eq:all-h-envs} $\Rightarrow$ \eqref{eq:smalltw-h-envs} $\Rightarrow$ \eqref{eq:smalltw-upto-h-envs} because more and more restrictions are placed on the environments. It remains to prove the converse implications. 

For the implication~\eqref{eq:all-h-envs} $\Leftarrow$ \eqref{eq:smalltw-h-envs}, we observe that the complement of $S$ is  an ideal with respect to parallel composition, i.e.~
\begin{align*}
G \not \in S \qquad \text{implies} \qquad G \oplus H \not \in S \qquad \text{for every $H$}	.
\end{align*}
This is  because  removing vertices and hyperedges can only make treewidth go down (here it is important that we measure treewidth after forgetting the sources).
Since $L \subseteq S$, it follows that using environments outside $S$  in~\eqref{eq:all-h-envs} will automatically make the equivalence  true, because both sides will be outside $L$. 

To prove the implication~\eqref{eq:smalltw-h-envs} $\Leftarrow$ \eqref{eq:smalltw-upto-h-envs} we use the fact that $\approx$ is a congruence recognising $L$. Therefore, whether or not the equivalence in~\eqref{eq:smalltw-h-envs} holds depends only on the $\approx$-class of an environment, and thus it is enough to test the equivalences using one environment from each $\approx$-equivalence class.
\end{proof}

 We are now ready to finish the proof of the theorem.
For an arity $n \in \Nat$, consider the following function 
\begin{align*}
  h_n : \hmonadn \Sigma n \to \powerset(\Hh_n) \qquad G \mapsto \set{H \in \Hh_n : G \oplus H \in L}.
\end{align*}
Define $A$ to be the ranked set where the $n$-ary elements are the elements of the image of  $h_n$, and define 
\begin{align*}
  h : \hmonad \Sigma \to A
\end{align*}
to be the function that works as $h_n$ on arity $n$. This function is surjective by definition of $A$.  By Claim~\ref{claim:two-other-parallel}, sourced hypergraphs have the same images under $h$ if and only if they are equivalent under $\sim$. Therefore, $A$ can be equipped with a product operation so that it becomes an  $\hmonad$-algebra $\alga$, and $h$ becomes  becomes the syntactic homomorphism. The product operation in $\alga$  is computable as required by item 2 of Definition~\ref{def:computable-algebra}: given an expression $G \in \hmonad A$, compute some inverse image under $\hmonad h$, then apply the product operation in $\hmonad \Sigma$, and then compute the value under $h$. The following claim establishes item 1 in Definition~\ref{def:computable-algebra}, i.e.~that one can enumerate the  elements of $A$ on each arity, and thus finishes the proof of the theorem.

\begin{claim}\label{claim:computable-images}
	For every $n$, one can compute the image of $h_n$.
\end{claim}
\begin{proof}
 By the proof of Claim~\ref{claim:two-other-parallel}, the function $h_n$ produces the empty set on every argument outside $S_n$. Therefore, it remains to compute the image of $h_n$ on arguments from $S_n$. By definition, the image $h_n(G)$ is equal to   
a set $\Gamma \subseteq \Hh_n$  if and only if  
\begin{align*}
  \bigwedge_{H \in \Gamma} H \oplus G \in L \qquad \text{and} \qquad  \bigwedge_{H \in \Hh_n - \Gamma} H \oplus G \not \in L
\end{align*}
It is not hard to see that the above, seen as a  property of $G$, is definable in counting \mso. Since counting \mso has decidable satisfiability on hypergraphs of bounded treewidth, we can effectively check if there is some $G \in S_n$ which satisfies the above property.
\end{proof}
 \end{proof}

In the above theorem, we used the assumption on bounded treewidth to construct the syntactic algebra. The following example shows that this assumption is essential, i.e.~the syntactic algebra is not computable for some languages definable in counting \msoh that have unbounded treewidth. Naturally, the reason is that without bounds on treewidth, satisfiability of \msoh becomes undecidable.

\begin{myexample}\label{ex:uncomputable-syntactic} A language has a syntactic algebra with one nullary element (and also on all other arities) if and only if it is full or empty. Since checking if a sentence of counting \msoh is full or empty is undecidable, it follows that there is no uniform way of computing a syntactic algebra given a sentence of counting \msoh.   More care is required to produce a single sentence of counting \mso whose syntactic algebra is not computable; we do this below.

 Define a \emph{tiling system} to be a tuple
\begin{align*}
  \Tt = \tuple { \underbrace{C}_{\substack{\text{set of}\\ \text{colours}}}, \underbrace{H,V \subseteq C \times C}_{\text{horizontal constraint}},\underbrace{c_0 \in C}_{\text{corner colour}} }
\end{align*}
A \emph{solution} to the tiling system is a defined to be a coloured directed graph satisfying the properties described in the following picture:
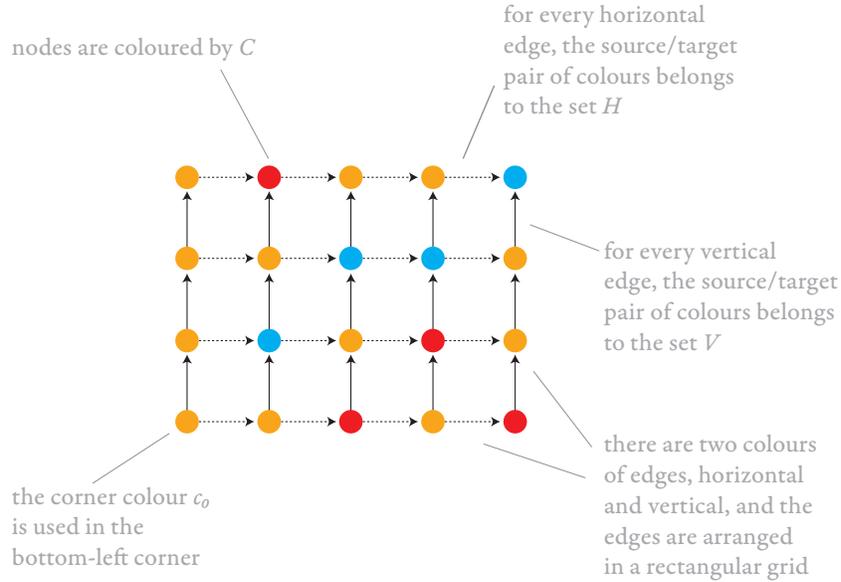
\mypic{46}
A solution can be viewed as a hypergraph over a ranked set that has  two binary labels (horizontal and vertical edges) and one unary label for each colour. Given a tiling system $\Tt$, it is straightforward to write a sentence of \msoh, in fact first-order logic, whose finite models are  exactly  the solutions of $\Tt$. 
Using a universal Turing machine, one can find a  tiling system $\Tt$, such that the following language is undecidable
\begin{align*}
  L = \set{ w \in C^+ : \text{in some solution, $w$ is an infix of the labels in some row}}
\end{align*} 
Consider the  set of solutions of the tiling system $\Tt$, and let $\sim$ be the syntactic  congruence of this hypergraph language (i.e.~the equivalence relation which identified two sourced hypergraphs if the have the same image under the syntactic homomorphism). On every arity, there is a distinguished ``error'' equivalence class of $\sim$, which is represented for example by any sourced hypergraph where some vertex has two outgoing horizontal edges. For a word $w \in C^+$ of length $n$, consider an $n$-ary  sourced hypergraph $G_w$ described in the following picture:
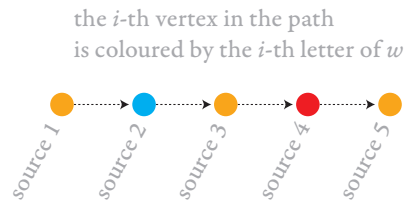
\mypic{47}
It is not hard to see that $w \in L$ if and only if the $\sim$-equivalence class of $G_w$ is not the ``error'' equivalence class. Since $L$ has undecidable membership, it follows that $\sim$ is undecidable.
\end{myexample}

\paragraph*{Testing aperiodicity}  Recall  that we have not yet proved Theorem~\ref{thm:decide-mso-not-count}, which says  that  the following problem is decidable.
	\begin{itemize}
		\item {\bf Input.} A number $k \in \Nat$ and a sentence $\varphi$  of counting \msoh.
		\item {\bf Question.} Is the following language definable in   \msoh without counting?
		\begin{align*}
\set{ G  : \text{$G$ is an undirected graph of  treewidth $\le k$ that satisfies $\varphi$}}
\end{align*}
	\end{itemize}
	We finish this section by proving the above  theorem.
	
\begin{proof}[Proof of Theorem~\ref{thm:decide-mso-not-count}]
Let $L$ be the language in the question. 
	By Theorem~\ref{thm:aperio}, $L$ is definable in \msoh without counting if and only if it is recognised by a $\hmonad$-algebra that is aperiodic and finite on every arity. Actually, by the remark in Footnote~\ref{foot:use-small-ranks}, one can compute some $n_0 \in \Nat$ such that $L$ is definable in \msoh without counting if and only if it is recognised by some algebra which is
	\begin{itemize}
		\item [(*)] finite on every arity and aperiodic on arities $\set{0,1,\ldots,n_0}.$
	\end{itemize}
	Since property (*) is preserved under surjective homomorphic images, it follows that $L$ is recognised by some algebra satisfying (*) if and only if the syntactic algebra of $L$ satisfies (*).  Apply Theorem~\ref{thm:compute-syntactic-algebra}, yielding a representation of the syntactic algebra, and then test if (*) holds by computing the multiplication tables for the semigroups corresponding to $\oplus$ on arities $\set{0,1,\ldots,n_0}$.	
\end{proof}

\section{Vertex replacement}
\label{sec:vertex}
This section is about a second monad for graphs,  called  $\vmonad$, which stands  for \emph{vertex replacement}, and is inspired by Courcelle's \vr-algebras. The structure of this section is meant to be parallel to the sections on $\hmonad$, i.e.~after defining the monad, we show that: its notion of recognisability coincides with the notion of \vr-recognisability that inspires the monad $\vmonad$; finite sets of polynomial operations generated languages of bounded width (in the case of the monad $\vmonad$, the appropriate notion of width is cliquewidth); and languages definable in counting \mso are recognisable. Since all of these results are existing results about \vr-recognisable languages that are on rephrased in the language of monads, and we have already discussed the monad approach to graphs on the example of the monad $\hmonad$, the proofs in this section are mainly rough sketches. 

Also, we do not include a discussion of computable algebras. Although analogues of Theorems~\ref{thm:computable-courcelle} (every language defined in counting \mso is recognised by a computable algebra) and~\ref{thm:compute-syntactic-algebra} (if the language furthermore has bounded width, in this case cliquewidth, then its syntactic  algebra can be computed), we do not include the proofs, mainly because we do not have a good application. (The application for $\hmonad$, i.e.~that \mso without counting corresponds to aperiodicity, was based on a result about \mso transductions computing tree decompositions, Theorem~\ref{thm:mso-transduction-treedecomp}, which is not known to hold for clique decompositions.)

\paragraph*{$\vmonad$-hypergraphs.} We begin by describing the notion of hypergraph used for the monad $\vmonad$.
In the hypergraphs for the monad $\hmonad$, arities were associated to hyperedges. For the monad $\vmonad$, arities will be associated to  vertices -- hence we use the name  hypervertices -- while edges will be binary, like in a directed graph. 

The monad $\vmonad$, like $\hmonad$, also uses ranked sets but, unlike  $\hmonad$, zero arities are not  allowed. Throughout this section, by \emph{ranked set} we mean a set were every element is assigned a nonzero natural number called its \emph{arity}. 
Define a \emph{corner} of a ranked set to be an element $v$ of the ranked set, together with a distinguished $i \in \set{1,\ldots,\text{arity of $v$}}$. We write $v[i]$ for such a corner.

\begin{definition}[$\vmonad$-hypergraphs]\label{def:v-hypergraphs}
A $\vmonad$-hypergraph consists of:
	\begin{enumerate}
			\item a nonempty ranked set $V$ of \emph{hypervertices};
			\item a ranked set $\Sigma$ of labels and a rank preserving labelling $V \to \Sigma$;
		\item a binary edge relation on the corners of $V$.
	\end{enumerate}
\end{definition}

The binary relation in item 3 is meant to represent the graph structure.  If $\Sigma$ has  only one unary label ``vertex'', then the hypervertices and their corners are the same thing, and   a $\vmonad$-hypergraph   is the same as a directed graph. Ultimately, the set $\vertvoc$ is the  set of labels that we care about. Other kinds of labels, and hypervertices with arities other than 1, are only used to define the monad structure. 

Here is how we draw $\vmonad$-hypergraphs:
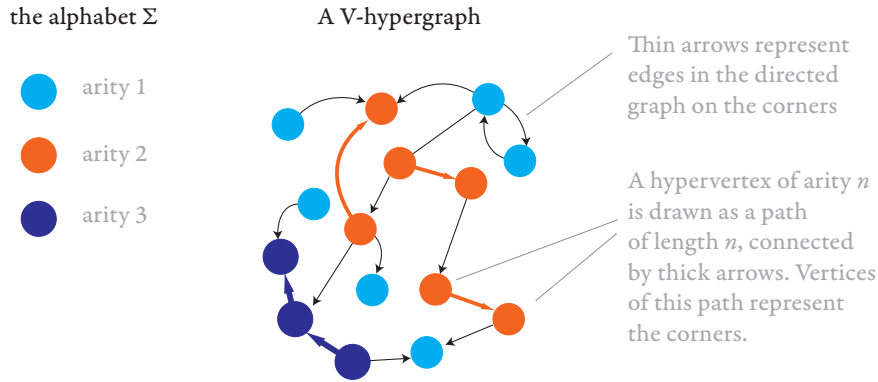
\mypic{48}

To define the monad structure, we  add \emph{ports} to $\vmonad$-hypergraphs. These are like the sources of  the monad $\hmonad$, except that instead of distinguishing $n$ vertices, we  group the vertices (actually, corners) into $n$ groups.

\begin{definition}[$ \vmonad$-hypergraph with ports]\label{def:sourced-v-hypergraphs}
For $n \in \set{1,2,\ldots}$, define an $n$-ary $\vmonad$-hypergraph with ports to be a $\vmonad$-hypergraph with hypervertices $V$, together with a \emph{port} function, not necessarily surjective, from its corners to  $  \set{1,\ldots,n}$. We write $i$-ports for corners with value $i$ under the source function.
\end{definition}

A 1-ary $\vmonad$-hypergraph with ports is the same thing as $\vmonad$-hypergraphs without port information, and hence arity 1 will play the same role -- ignoring the port/source annotation -- for the monad $\vmonad$ as was played by arity 0 for the monad $\hmonad$.
Here is how we draw $\vmonad$-hypergraphs with ports. 
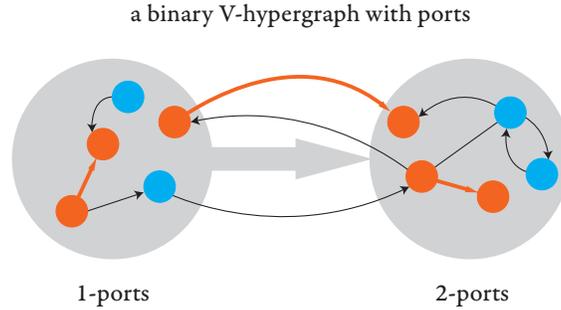
\mypic{25}

 The point of the drawing above  is to underline that a   $\vmonad$-hypergraph with ports can be itself used as a label for a hypervertex.
We now define the monad structure on sourced $\vmonad$-hypergraphs.
\begin{definition}[The monad $\vmonad$]\label{def:vmonad}Define $\vmonad$ to be the following monad. 
\begin{itemize}
	\item \emph{Category.} The category is ranked sets and arity-preserving functions, with the arities being $\set{1,2,\ldots}$.
	\item \emph{On objects.} For a ranked set $\Sigma$,    $\vmonad \Sigma$ is  the ranked set of finite\footnote{ Finitely many hypervertices (and therefore finitely many corners, since hypervertices have finite arities).
}  $\vmonad$-hypergraphs with ports,  using labels $\Sigma$,  modulo isomorphism\footnote{An isomorphism is a bijection of the vertices which preserves all of the information in   Definitions~\ref{def:v-hypergraphs} and~\ref{def:sourced-v-hypergraphs}.}.
\item \emph{On morphisms.} For an arity-preserving function $f : \Sigma \to \Gamma$,  the function $\vmonad f : \vmonad \Sigma \to \vmonad \Gamma$  changes the labelling of its input according to  $f$.
\item \emph{Unit.}  The unit of  an $n$-ary letter $a \in \Sigma$ is defined to be the $n$-ary  $\vmonad$-hypergraph with ports which has one $n$-ary hypervertex $v$ labeled by $a$, and where the port function is $v[i] \mapsto i$, as in the following picture:
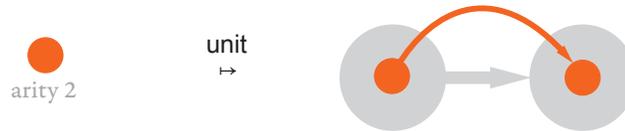
\mypic{30}
\item \emph{Flattening.} The flattening of $G \in \vmonad \vmonad \Sigma$ is defined as follows, see Figure~\ref{fig:vmonad-flattening}. 
\begin{itemize}
\item \emph{Hypervertices and their labels.} Hypervertices are pairs $(v,w)$ such that $v$ is a hypervertex of $G$ and $w$ is a hypervertex in the label of $v$. The arity and label are inherited from $w$.
\item \emph{Ports.}  The number of ports is inherited from $G$. If $w[i]$ is a $j$-port in the label of $v$, and $v[j]$ is a $k$-port in $G$, then $(v,w)[i]$ is a $k$-port in the flattening.
	\item \emph{Edges.}  The flattening has an edge 
	\begin{align*}
(v,w)[i] \to (v',w')[i']
\end{align*}
if one of the two conditions is satisfied ($j$ is defined to be the port number of $w[i]$ in the label of $v$, likewise for $j'$):
\begin{align*}
v=v' \text{ and }w[i]\to w'[i'] \qquad \text{or} \qquad v[j] \to v'[j']	
\end{align*}
\end{itemize}
 \end{itemize}
\end{definition}

\begin{figure}
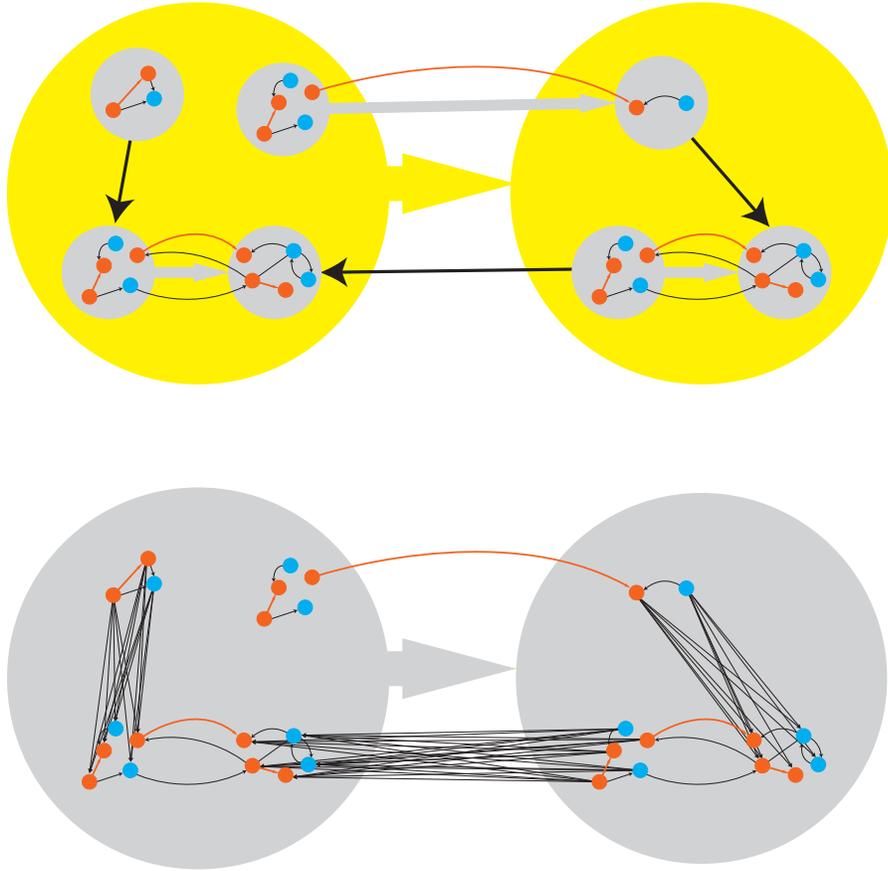

\mypic{26}
\mypic{31}
  \caption{\label{fig:vmonad-flattening}A binary $\vmonad$-hypergraph (above) and its flattening (below)}
\end{figure}

\begin{fact}
	$\vmonad$ satisfies the axioms of  a monad.
\end{fact}
\begin{proof} A routine check. This proof is actually a bit easier than the one for $\hmonad$, since there is not deleting involved (in the monad $\hmonad$, the source vertices from the nested hypergraphs were deleted).
We only prove that flattening is associative. 
	Let  $G \in \vmonad \vmonad \vmonad \Sigma$. 
Define $G_1,G_2 \in \vmonad \Sigma$ to  be the results of applying to $G$ the following functions (to make the comparison easier, we discuss the two $\vmonad$-hypergraphs in parallel columns for the rest of this proof):
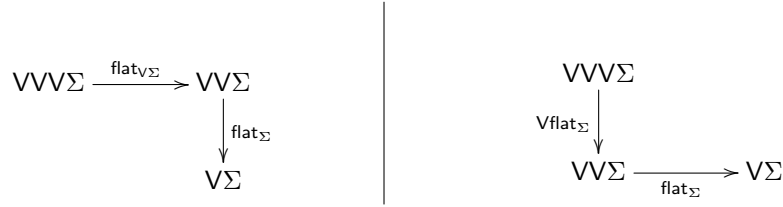
\petrisan{
\begin{align*}
\xymatrix @R=2pc @C=3pc {\vmonad \vmonad \vmonad \Sigma  \ar[r]^{\flatt_{ \vmonad \Sigma}}  & \vmonad \vmonad \Sigma \ar[d]^{\flatt_\Sigma}  \\
		& \vmonad \Sigma
		}	
\end{align*}
}{
\begin{align*}
	\xymatrix @R=2pc @C=3pc {\vmonad \vmonad \vmonad \Sigma   \ar[d]_{\vmonad{ \flatt_\Sigma}} &   \\
		\vmonad \vmonad \Sigma \ar[r]_{ \flatt_\Sigma}& \vmonad \Sigma
		}
\end{align*}
}
Our goal is to show that $G_1$ is the same (isomorphism type)  as  $G_2$. 
	\petrisan{Hypervertices in $G_1$ are of the form
	\begin{align*}
((v,w),u)
\end{align*}
where:
 \begin{itemize}
	\item $v$ is a hypervertex of $G$;
	\item $w$ is a hypervertex in the label of~$v$;
	\item $u$ is a hypervertex in the label of~$w$;
\end{itemize}
} {Hypervertices in $G_2$ are of the form
	\begin{align*}
(v,(w,u))
\end{align*}
where the conditions on $v,w,u$ are the same as for $G_1$.} The isomorphism between $G_1$ and $G_2$ is going to be
\begin{align*}
((v,w),u) \mapsto (v,(w,u)).
\end{align*}
  The function above is a bijection, since the conditions on $v,w,u$ are the same on both sides. The function also  preserves  labels and arities, because  labels and arities are inherited from $u$ on both sides. Let us now look at the port mapping. Consider corners on both sides:
\petrisan{
\begin{align*}
((v,w),u)[i]	
\end{align*}
}
{
\begin{align*}
(v,(w,u))[i]	
\end{align*}
}
Choose $j,k,l$ so that $u[i]$ is a $j$-port in the label of $w$, and $w[j]$ is a $k$-port in the label of  $v$, and $v[k]$ is an $l$-port in the graph $G$. By unravelling the definition of flattening, it follows that both corners discussed above are going to be  $l$-ports. To complete the proof, we need to show that the bijection preserves edges, i.e.~the following conditions are equivalent:
\petrisan{
 $G_1$ has an edge
\begin{align*}
((v,w),u)[i] \to ((v',w'),u')[i']
\end{align*}
}
{
 $G_2$ has an edge
\begin{align*}
(v,(w,u))[i] \to (v',(w',u'))[i']
\end{align*}
}
Define $j,k,l$  as in the dicussion of the port functions for $v,w,u$, likewise define $j',k',l'$ for $v',w',u'$. By unravelling the definition of edges in the flattening, we see that the existence of both edges described is equivalent to the following disjunction:
\begin{itemize}
	\item $v=v'$ and $w=w'$ and there is an edge $u[i] \to u'[i']$ in the label of $w$; or
	\item $v=v'$ and there is an edge $w[j] \to w'[j']$ in the label of $v$; or
	\item there is an edge  $v[k] \to v'[k']$ in $G$.
\end{itemize}
\end{proof}

\paragraph*{Recognisable languages.} As mentioned after the definition of ports, having one port is the same as having no port information. Therefore, a language of $\vmonad$-hypergraphs (without port information) can be identified with a set $L \subseteq \vmonadone \Sigma$. A language is called recognisable if it is recognised by a homomorphism from the entire algebra $\vmonad \Sigma$ into an algebra that is ``finite'' in some suitable sense. For the same reasons as in the monad $\hmonad$, the notion of ``finite'' algebra is chosen to mean ``finite on every arity''.
\begin{definition}[Recognisable language]\label{def:v-recognisable}
	A language $L \subseteq \vmonadone \Sigma$ is called \emph{recognisable} if it is recognised by a $\vmonad$-algebra that is finite on every arity.
\end{definition}

We now show  that recognisability, as defined above,  coincides with Courcelle's definition of \vr-recognisability, which is the inspiration for the monad $\vmonad$. We proceed the same way as for the monad $\hmonad$, i.e.~we distinguish a subset of polynomial operations in algebra $\vmonad \Sigma$, and then define \vr-recognisability in terms of that subset. For the monad $\hmonad$, the main role was played by parallel composition; for $\vmonad$ the main role will be played by disjoint union $\oplus$, which is defined in the natural way, as illustrated in the following picture:
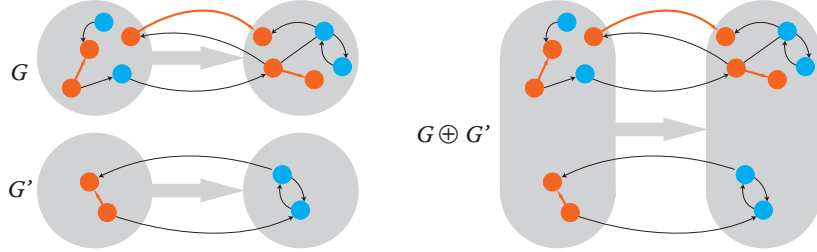
\mypic{49}
In general, the inputs of disjoint union might have different arities, the arity of the output is then defined to be maximal arity of the inputs. It is convenient here that the port function need not be surjective, this way we can use the disjoint union for arguments that have pairwise disjoint sets of port numbers that are used.

We are now ready to present Courcelle's notion of \vr-recognisability. 

\begin{definition}[\vr-recognisable language]\myfootcite[Definition 4.52 and Theorem 4.57]{Courcelle:2012wq}
	\label{def:vr-recognisable}A language $L \subseteq \vmonadone \Sigma$ is called \emph{\vr-recognisable} if there is an  equivalence relation on $\vmonad \Sigma$, which:
	\begin{enumerate}
		\item recognises $L$, i.e.~$L$ is union of equivalence classes; and
		\item has finitely many equivalence classes on every arity; and
		\item is compatible with all of the \vr-operations defined in Figure~\ref{fig:vr-operations}.
	\end{enumerate}
\end{definition}

The notions of recognisability from  Definitions~\ref{def:v-recognisable} and~\ref{def:vr-recognisable} coincide. 
\begin{theorem}\label{thm:same-as-vr}
	For a language $L \subseteq \vmonadone{\Sigma}$, the following conditions are equivalent:
	\begin{enumerate}
		\item $L$ is recognisable;
		\item $L$ is recognised by a congruence (in the sense of Definition~\ref{def:congruence}) with finitely many equivalence classes on every arity;
		\item $L$ is  \vr-recognisable.
	\end{enumerate}
\end{theorem}
\begin{proof}
Same kind of proof as for  Theorem~\ref{thm:same-as-hr} in the monad $\hmonad$. The equivalence of 1 and 2 holds for every monad in a category of sorted sets, and the implication from 2 to 3 holds because all of the \vr-operations are polynomial operations. To prove the implication from 3 to 2,  the key observation is that the \vr-operations are enough to construct all linear unary polynomial operations.
\end{proof}

\begin{figure}
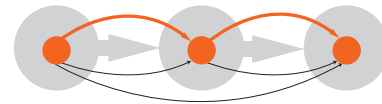

\petrisan{	For every arity $n \in \set{1,2,\ldots}$, there is an operation  which inputs two $n$-ary  $\vmonad$-hypergraphs, and outputs their disjoint union.}{\ \vspace{-.6cm}
\mypic{52}}

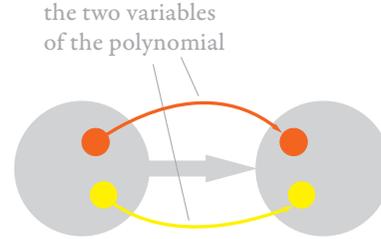
\petrisan{	There is a constant  for every unit.}{Constants are polynomial operations.}

\petrisan{For every function     
\begin{align*}
f : \set{1,\ldots,n} \to \set{1,\ldots,m}	
\end{align*}
 with $n,m \in \set{1,2,\ldots}$, there is an operation which inputs an $n$-ary $\vmonad$-hypergraph, and outputs an $m$-ary one where the ports are updated along $f$. }
{To see that this operation is a polynomial operation, consider  
\begin{align*}
f : \set{1,2,3,4} \to \set{1,2} 
\end{align*}
defined by $(1,2,3,4) \mapsto (1,2,2,1)$. Then the polynomial operation corresponding to this operation looks like this:
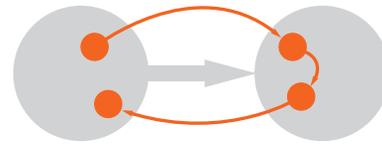
\mypic{53}}

\petrisan{For every  \begin{align*}
E \subseteq \set{1,\ldots,n}^2	
\end{align*}with $n \in \set{1,2,\ldots}$ there is an operation which inputs an $n$-ary $\vmonad$-hypergraph and outputs one of the same arity, where 
which adds  (directed) edges connecting all $i$-th sources to all $j$-th sources for all $(i,j) \in E$.}{To see that this operation is a polynomial operation, suppose that $n=3$ and 
\begin{align*}
E = \set{(1,2),(2,3),(1,3)}.
\end{align*}
Then the polynomial operation looks like this:
 \mypic{54}}

  \caption{\label{fig:vr-operations}The \vr-operations (left column) and why they are polynomial operations (right column).}
\end{figure}


\paragraph*{Cliquewidth.} As the monad $\hmonad$ was to treewidth, the monad $\vmonad$ is to cliquewidth. The latter is a graph parameter, which is more generous than treewidth in the following sense: bounded treewidth implies bounded cliquewidth, but not the other way round (as the name implies, the class of cliques has bounded cliquewidth, in fact cliquewidth one, while its treewidth is unbounded). To relate cliquewidth with the monad $\vmonad$, not much work needs to be done, since already the definition of cliquewidth  is in terms of the \vr-operations.

\begin{definition}[Cliquewidth]
\myfootcite[Definition 2.89]{Courcelle:2012wq}
The \emph{cliquewidth} of $G \in \vmonad \Sigma$ is defined to be the smallest $n \in \set{1,2,\ldots}$ such that $G$ can be generated  using the \vr-operations that use only arities $\le n$.
\end{definition}

We are mainly interested in the  case of the cliquewidth of directed graphs, i.e.~when $G$ has arity one and  $\Sigma$ has only one unary label. 

\begin{theorem}
	A set $L \subseteq \vmonad \Sigma$ has bounded cliquewidth if and only if it is contained in a set  generated by finitely many polynomial operations of the $\vmonad$-algebra $\vmonad \Sigma$.
\end{theorem}
\begin{proof}
The left-to-right implication follows immediately from the definition of cliquewidth and the following observations:  all \vr-operations are polynomial operations, and  there are finitely many \vr-operations using a fixed finite set of arities. The converse  implication follows from the following claim.
\begin{claim}
	For every ranked set  $X$ and polynomial  operation
	\begin{align*}
  p : (\vmonad \Sigma)^X \to \vmonad \Sigma
\end{align*}
there  exists some $k \in \Nat$ such that for every valuation $\eta \in  (\vmonad \Sigma)^X$ satisfies
\begin{align*}
  \text{cliquewidth of $p(\eta)$} \quad \le \quad k + \max_{x \in X} \text{cliquewidth of $\eta(x)$}
\end{align*}
\end{claim}
\begin{proof}
An upper bound on $k$ is the number of corners in the  $\vmonad$-hypergraph defining the polynomial operation.
\end{proof}
\end{proof}

\paragraph{Monadic second-order logic.} In this section we prove a variant of Courcelle's Theorem for the monad $\vmonad$, which says that languages definable in counting \mso are necessarily recognisable. At this point, we can afford a short proof of the theorem, due to  the generic form of the proof of Courcelle's Theorem in Section~\ref{sec:courcelle-theorem}.
The idea of for definability in \mso is the same as in Section~\ref{sec:mso}, i.e.~we associate to each element of $\vmonad \Sigma$ a model, and then use logic -- mainly counting \mso\ -- to describe properties of that model.

\begin{definition}\label{def:v-model}
	For $G \in \vmonad \Sigma$, define its \emph{model} $\structv G$ as follows:
\begin{enumerate}
	\item the universe is the corners of the hypervertices;
	\item there is a binary relation for the edges between the corners;
	\item for every $a \in \Sigma$ and $i \in \set{1,\ldots,\text{arity of $a$}}$ there is a unary relation selecting  $i$-th corners of hypervertices with label $a$;
\item for every $i \in \set{1,\ldots,\text{arity of $G$}}$ there is unary relation for the $i$-ports. 
\end{enumerate}
\end{definition}
The vocabulary of the model depends on the number of ports, because of the relations in item 4. However, if the alphabet $\Sigma$ is finite and the arity is fixed, then the vocabulary is finite.

Consider the special case when $G \in \vmonadone \vertvoc$, i.e.~$G$ has one port (which is like having no port information) and the set of labels has only one unary label. In other words, $G$ is a directed graph (modulo isomorphism).  The model $\structv G$ has the vertices as the universe and there is a binary edge relation; the other relations defined in items 3 and 4 are meaningless  because they select all elements in the universe. In this case, the model $\structv G$ corresponds to the encoding $\lfloor G \rfloor$ of directed graphs that was discussed in Example~\ref{ex:digraphs-h}. Traditionally, \mso over such encodings of graphs (i.e.~the universe is only the vertices, and not the edges) is referred to as \msov, which motivates the following definition.

\begin{definition}[Language definable in counting \msov]
Let $\Sigma$ be a finite ranked set.	A language $L \subseteq \vmonadone \Sigma$ is called \emph{definable in counting \msov} if there is a sentence $\varphi$ of counting \msov which defines it in the following sense:
	\begin{align*}
  G \in L \quad \text{iff} \quad \structv G \models \varphi \qquad \text{for every } G \in \vmonadone \Sigma.
\end{align*}

\end{definition}
Here is the $\vmonad$ version of Courcelle's Theorem.

\begin{theorem} \myfootcite[Theorem 5.68 (1)]{Courcelle:2012wq}	If a language $G \subseteq \vmonadone \Sigma$ is definable in counting \msov, then it is recognisable. 
\end{theorem}
\begin{proof} Recall the proof of Courcelle's Theorem for the monad $\hmonad$ from Section~\ref{sec:courcelle-theorem}. The only part of proof which depended on the choice of monad was the Compositional Encoding Lemma. Therefore, to prove the theorem, it suffices to show the analogous result for the monad $\vmonad$, which we do below.
	\begin{lemma}[Compositional encoding]\label{lem:vgood-encoding} For a ranked set $\Sigma$, define $\structv \Sigma$ to be  the image of the set $\vmonad \Sigma$ under the function $\structv$; this image is a viewed as a ranked set. For every ranked set 
$X$ and every  $G \in \vmonad X$ there  is an operation $f$ compatible with counting \mso which makes the following diagram commute:
\begin{align*}
\xymatrix@C=4cm{
(\vmonad \Sigma)^X \ar[r]^{(\structv)^X} \ar[d]_{\sem G} &  (\structv \Sigma)^X \ar[d]^{f}\\
\vmonad \Sigma \ar[r]_{\structv\ } &  \structv \Sigma
}
\end{align*}

\end{lemma}
\begin{proof}[of the lemma]
Let $\eta \in (\vmonad \Sigma)^X$ be a valuation of the variables, and define  $\structa_x$ to be the model that  represents $\eta(x)$. Our goal is to transform  the models $\set{\structa_x}_{x \in X}$ into the model  that representing $\sem G (\eta)$,  using operations compatible with counting \mso. We simply take the disjoint union
\begin{align*}
\coprod_{v \in \text{hypervertices of $G$}} \structa_{\text{label of $v$}}	
\end{align*}
and observe that the structure defined in the flattening (labels, edges, sources) can all be defined in a quantifier-free way. For the monad $\vmonad$, unlike for $\hmonad$, we do not need to use quantifier-free universe restrictions.
\end{proof}
\end{proof}

Without any additional assumptions, the converse of the above theorem is false, as shown in  the following example. 

\begin{myexample}\label{ex:v-grids}
For a set $X \subseteq \set{1,2,\ldots}$, consider the set of  undirected $n \times n$ grids with $n \in X$. This language is not definable in counting \msov for some choices of $X$, e.g.~when $X$ is undecidable. Nevertheless, one can show\footnote{\cite[Proposition 4.36]{Courcelle:2012wq} gives the result for $\hmonad$, but the proof is similar for $\vmonad$} that this language is recognisable regardless of choice of $X$.
\end{myexample}

The languages in the above example used unbounded cliquewidth. This raises the question: if we additionally assume bounded cliquewidth, is recognisability equivalent to definability in counting \msov? In other words, is Theorem~\ref{thm:courcelle-conjecture} true for $\vmonad$-hypergraphs? This is an open problem, although a special case has been proved recently\myfootcite[Theorem 3]{DBLP:journals/corr/abs-1803-05937}, namely recognisability is equivalent to definability in counting \msov under the stronger assumption of bounded \emph{linear cliquewidth}. One definition of bounded linear cliquewidth is that a language is generated by a finite set of linear unary polynomial operations.

\bibliographystyle{unsrtnat}
\pagebreak
\bibliography{bib}
\end{document}